\documentclass[aps,pra,reprint,twocolumn,superscriptaddress,nofootinbib,longbibliography]{revtex4-1}

\usepackage{graphicx}
\usepackage{bm}
\usepackage{amsmath,amssymb,amsfonts,amsthm,mathtools,empheq}
\usepackage{xcolor}
\usepackage{hyperref}
\usepackage[capitalize]{cleveref}
\usepackage[normalem]{ulem}
\usepackage{wrapfig}
\usepackage{dsfont}
\usepackage{mathrsfs,esvect,tensor}
\usepackage{etoolbox}
\usepackage{soul}

\newtheorem*{theorem*}{Theorem}
\newtheorem{theorem}{Theorem}
\newtheorem{lemma}{Lemma}
\hypersetup{colorlinks=true,urlcolor=[rgb]{0,0,0.5},citecolor=[rgb]{0.5,0,0},linkcolor=[rgb]{0,0,0.4}}

\usepackage{pgffor}
\foreach \x in {A,...,Z}{%
\expandafter\xdef\csname c\x\endcsname{\noexpand\ensuremath{\noexpand\mathcal{\x}}}
}
\foreach \x in {A,...,Z}{%
\expandafter\xdef\csname f\x\endcsname{\noexpand\ensuremath{\noexpand\mathscr{\x}}}
}
\foreach \x in {A,...,Z}{%
\expandafter\xdef\csname b\x\endcsname{\noexpand\ensuremath{\noexpand\mathbb{\x}}}
}

\newcommand{\rank}[0]{\mathrm{rank}}
\newcommand{\imgbox}[2]{\vcenter{\hbox{\includegraphics[width=#1]{#2}}}}
\newcommand{\ket}[1]{\ensuremath{|{#1}\rangle\!}}
\newcommand{\bra}[1]{\ensuremath{\langle{#1}|}}

\newcommand{\braket}[1]{\ensuremath{\langle{#1}\rangle}}
\def\tr{{\rm tr}}

\DeclareMathOperator*{\ex}{\mathbb{E}}

\DeclareMathOperator*{\var}{\text{Var}}

\newcommand{\idm}[0]{\mathds{1}}

\newcommand{\cinf}[0]{\overline{\text{CI}}}

\newcommand{\aotv}[0]{\protect\vv{\textsf{AOT}}_{\!t,x}}

\DeclarePairedDelimiter\abs{\lvert}{\rvert}%
\DeclarePairedDelimiter\norm{\lVert}{\rVert}%

\makeatletter
\let\oldabs\abs
\def\abs{\@ifstar{\oldabs}{\oldabs*}}
\let\oldnorm\norm
\def\norm{\@ifstar{\oldnorm}{\oldnorm*}}
\makeatother

\makeatletter
  \let\oldparagraph\paragraph
  \def\paragraph{\@ifstar{\@sparagraph}{\@uparagraph}}
  \def\@sparagraph#1{\oldparagraph*{#1}\nobreak\hskip -0.666em}
  \def\@uparagraph#1{\oldparagraph{#1}\nobreak\hskip -0.666em}
\makeatother

\newcommand{\anc}[0]{\mathrm{anc}}

\newcommand{\phys}[0]{\mathrm{phys}}

\newcommand{\mgate}{\mathbin{\includegraphics[height=1.1ex]{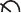}}}
\newcommand*\widefbox[1]{\fbox{\hspace{2em}#1\hspace{2em}}}

\makeatletter
\newcommand{\printappendixtoc}[1][Appendix Contents]{%
  \begingroup
    \let\saved@addcontentsline\addcontentsline
    \let\addcontentsline\@gobblethree
    \section*{#1}%
    \let\addcontentsline\saved@addcontentsline
  \endgroup
  \@starttoc{atoc}%
}

\newcommand{\appendixonlytocsetup}{%
  \let\orig@addcontentsline\addcontentsline
  \renewcommand{\addcontentsline}[3]{%
    \def\temp@ext{##1}\def\tocext{toc}%
    \ifx\temp@ext\tocext
      \orig@addcontentsline{atoc}{##2}{##3}%
    \else
      \orig@addcontentsline{##1}{##2}{##3}%
    \fi
  }%
}
\makeatother
\makeatletter%
\renewcommand\onecolumngrid{% <<<<<<
\do@columngrid{one}{\@ne}%
\def\set@footnotewidth{\onecolumngrid}% <<<<<<<<<<<<<<<<
\def\footnoterule{\kern-6pt\hrule width 1.5in\kern6pt}%
}%
\renewcommand\twocolumngrid{% <<<<<<
        \def\footnoterule{% restore rule
        \dimen@\skip\footins\divide\dimen@\thr@@
        \kern-\dimen@\hrule width.5in\kern\dimen@}
        \do@columngrid{mlt}{\tw@}
}%
\makeatother%

\begin{document}

\title{Local arrows of time in quantum many-body systems}
\author{Andrew G. Yates}
\email{andrewyates@g.harvard.edu}
\affiliation{\it Department of Physics, Harvard University, Cambridge, MA 02138, USA}
\author{Jordan Cotler}
%\email{jcotler@fas.harvard.edu}
\affiliation{\it Department of Physics, Harvard University, Cambridge, MA 02138, USA}
\author{Nishad Maskara}
%\email{nmaskara@g.harvard.edu}
\affiliation{\it Department of Physics, Harvard University, Cambridge, MA 02138, USA}
\author{Mikhail D.~Lukin}
%\email{lukin@physics.harvard.edu}
\affiliation{\it Department of Physics, Harvard University, Cambridge, MA 02138, USA}

\begin{abstract}
We demonstrate that in quantum many-body systems, local arrows of time can differ from the global time $t$ induced by Hamiltonian evolution.  That is, within a quantum many-body system, the flow of time can be relative to each observer or by proxy each local subsystem.  We provide a definition of local arrows of time in quantum many-body systems, and explain their relation to spacetime quantum entropies.  Then we give a variety of numerical and analytical examples which explore different ways in which local arrows of time can manifest in quantum many-body dynamics, including exotic arrows of time arising from quantum thermalization and quantum error correction.  
\end{abstract}

\maketitle

%--------------------------------------

\section{Introduction}

One of the foundational insights of special and general relativity is that time is relative to an observer.  That is, even if a fiducial global time $t$ is provided for the dynamics of a relativistic system, the passage of time experienced by an observer may differ from $t$.  Interestingly, in quantum many-body systems there is a related disparity between global time and local time.  The global time $t$ is the parameter appearing in the Hamiltonian time evolution operator $e^{- i H t}$, but it may disagree with the passage of time for an observer living within the system, or a proxy local subsystem~\cite{cotlerQuantumCausalInfluence2019a}.  As such, how do we understand and analyze the local flow of time for a local subsystem within a quantum many-body system?

In this paper we investigate the above question using tools from quantum information theory, and define local arrows of time within quantum many-body systems.  We show that these arrows of time comport with many of our intuitions about how time should be `experienced' within a quantum many-body system, but also bear some initially counter-intuitive features mediated by the presence of quantum entanglement.  These features build upon and complement previous studies of time within quantum many-body systems~\cite{connesNeumannAlgebraAutomorphisms1994, dibiagioArrowTimeOperational2021, fitzsimonsQuantumCorrelationsWhich2015, macconeQuantumSolutionArrowoftime2009, pageEvolutionEvolutionDynamics1983, liuInferringArrowTime2024, cotlerSuperdensityOperatorsSpacetime2018}, and elucidate the role of thermalization and quantum error correction in realizing exotic arrows of time.

This paper is organized as follows.  First we define local arrows of time in quantum many-body systems on a one--dimensional lattice.  Next, we show how the local arrows of time relate to spacetime quantum entropy~\cite{bou-comasMeasuringTemporalEntanglement2024, doiTimelikeEntanglementEntropy2023, leroseScalingTemporalEntanglement2021, milekhinObservableComputableEntanglement2025, sonnerInfluenceFunctionalManybody2021}, and involve a tradeoff between the growth of entropy in a quantum state and certain Hamiltonian correlators.  Thereafter we provide a number of numerical and analytical examples which highlight local arrows of time in quantum many-body systems, and explain connections to thermalization and quantum error correction.  Next we explain how local arrows of time can be measured in contemporary quantum simulators~\cite{maskaraProgrammableSimulationsMolecules2025, xuConstantoverheadFaulttolerantQuantum2024}.  We conclude with a discussion of our results and future directions.
%--------------------------------------
\begin{figure}[h!]
    \centering
    \includegraphics[]{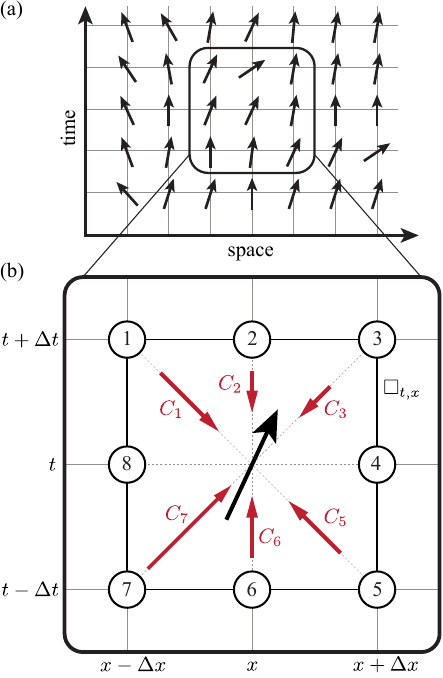}
    \caption{(a) A depiction of a spacetime lattice. (b) The inlay shows the calculation of the local arrow of time $\aotv$ (black arrow) as a sum of contributions from neighbors $q_i \in \square_{t,x}$ (Eq.~\ref{eq:MAIN_aot_def}), indexed $i=1,...,8$. Each dashed, red arrow represents the term corresponding to neighbor $q_i$, weighted by the causal influence $C_i=\cinf_{q_i,(t,x)}$ from $q_i$ to $(t,x)$. Only six red arrows are depicted because $q_4$ and $q_8$ are spacelike separated from $(t, x)$, and thus cannot contribute to the local arrow of time.}
    \label{fig:spacetimelattice}
\end{figure}
%--------------------------------------

\section{Local arrow of time}

The causal structure of general quantum many-body systems can be characterized by the quantum causal influence~\cite{cotlerQuantumCausalInfluence2019a} (QCI). The QCI, which we denote by $\overline{\text{CI}}_{AB}$, is a quantum information-theoretic quantity that diagnoses the causal relationship between two spacetime regions $A$ and $B$, telling us how much $A$ `affects' $B$. For instance, it may be that $A$ affects $B$ but $B$ does not affect $A$, and so generally $\overline{\text{CI}}_{AB} \not = \overline{\text{CI}}_{BA}$.

Although the QCI can be computed in very general quantum systems without any a priori notion of time (e.g.~quite general tensor networks)~\cite{cotlerSuperdensityOperatorsSpacetime2018, feynmanTheoryGeneralQuantum1963}, here we will consider the more specialized and physically salient setting of a quantum many-body system on a lattice with global time $t$ induced by Hamiltonian evolution.  Suppose that we have two spacetime subregions $R_A$ and $R_B$. Here $A$ will denote the collection of lattice sites at time $t_A$, and $B$ the collection of lattice sites at time $t_B$. If the state of the system at time $t = 0$ is denoted by $\ket{\Psi}$, the time evolution operator is denoted by $U(t) \equiv e^{- i H t}$, a unitary in $R_A$ is denoted by $V_{R_A} \equiv U^\dagger(t_A) V_{\!A} U(t_A)$, and a Hermitian operator in $R_B$ is denoted by $O_{R_B} \equiv U^\dagger(t_B) O_{B} U(t_B)$, then $\overline{\text{CI}}_{AB}$ can be written as~\cite{cotlerQuantumCausalInfluence2019a}
\begin{align}
\label{eq:cinfAB}
\overline{\text{CI}}_{AB} \equiv \ex_{O_{\!B}} \!\var_{V_{\!A}} \braket{\Psi| V_{R_A}^\dagger O_{R_B} V_{R_A} |\Psi}\,.
\end{align}
Here $\mathbb{E}_{O_B}$ takes the average over positive semidefinite operators $O_B$ drawn from the Hilbert-Schmidt measure~\cite{bengtsson_zyczkowski_2020} on $B$, and $\operatorname{Var}_{V_A}$ takes the variance over unitaries $V_A$ drawn from the Haar measure on $A$.  
More intuitively, Eq.~\eqref{eq:cinfAB} is an information-theoretic response function: it measures how much expectation values in $B$ react to a unitary ``kick'' in $A$.  
Consider how the expectation value $\langle O_{R_b} \rangle$ changes while sweeping $V_A$ over the Haar ensemble. 
If the expectation value is insensitive to $V_A$ for every observable $O_B$, then the variance vanishes and $\overline{\text{CI}}_{AB}=0$.
More generally, the magnitude of $\overline{\text{CI}}_{AB}$ quantifies how strongly perturbations in $A$ can be detected in $B$.
Importantly, the value of $\overline{\text{CI}}_{AB}$ depends on the state $|\Psi\rangle$, the Hamiltonian $H$, and the chosen subregions.

Our aim in this paper is to develop a \textit{local} probe of quantum causal structure, enabling the study of the \textit{local arrow of time} within a quantum many-body system.  To do so, we consider a quantum many-body system on an evenly-spaced 1D lattice with spacing $\Delta x$, although our methods readily generalize.  Considering the system at times $t = n\,\Delta t$ for $n$ integer, we construct a spacetime lattice with temporal lattice spacing $\Delta t$ (see \cref{fig:spacetimelattice}).

For a given point $(t,x)$ on the lattice, consider the set $\square_{t,x} = \{(t + a \Delta t, x + b \Delta x) \, : \, (a,b) \in \{-1,0,1\}^2\setminus (0,0) \}$.  That is, $\square_{t,x}$ is the set of all lattice points on the edges and vertices of the `box' surrounding $(t,x)$.  Letting $\vec{v}_{q,(t,x)}$ be a vector from a point $q$ on the lattice to $(t,x)$, the local arrow of time is a spacetime vector given by
\begin{align}\label{eq:MAIN_aot_def}
\vv{\textsf{AOT}}_{\!t,x} \equiv \sum_{q \in \square_{t,x}} \overline{\text{CI}}_{q,(t,x)} \vec{v}_{q,(t,x)}\,.
\end{align}
This quantity is the average flow of the quantum causal influence through the point $(t,x)$, and we will show below that it unveils novel causal structure hidden in Hamiltonian dynamics.

%--------------------------------------

%--------------------------------------
\begin{figure*}
    \centering
    \includegraphics[width = \textwidth]{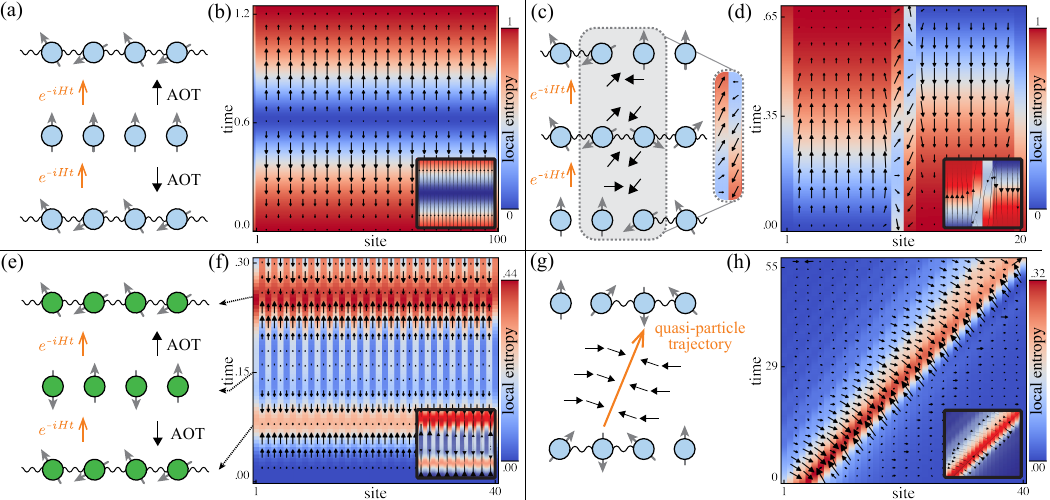}
    \caption{%
   \textbf{Vectorfield $\aotv$ under Hamiltonian $H$ with time step $\Delta t$ and an initial state $\ket{\Psi}$.} Heat maps show the single-site von Neumann entropy at each lattice site. In each subplot the evolution runs up to a maximum time $T$, which is $T = 1.2$, $0.65$, $0.30$, and $57$, respectively. All panels except (e) and (f) use the same Hamiltonian $H = \sum_j X_jX_{j+1} + .01\sum_j X_j - .21\sum_j Z_j$ and $\Delta t = 0.005$, varying only $T$ and $\ket{\Psi}$. Panels (e) and (f) use the PXP Hamiltonian.
    (a, b) The high-entropy initial state $\ket{\Psi}$ is generated by evolving a product state backwards in time by $T/2$.  Evolving that state forward under $H$ then exactly retraces the backward step, producing the blue fringe at $t = T/2$.  The arrow of time vectorfield points away from this fringe, following the local entropy gradient.
    (c, d) The system begins from a product state whose right half has been evolved backward in time by $T/2$ and whose left half has been evolved forward by $T/2$.  Under subsequent evolution by $H$, the left side becomes more entangled while the right side becomes less entangled.  By $t = T$, each half’s entropy profile is exactly reversed from its initial profile at $t = 0$.  The interaction at their boundary (sites 10 and 11) then gives rise to spatiotemporal $\aotv$ vectors at the interface.
    (e, f) Simulating the PXP Hamiltonian, the system is initialized in a Néel state, and displays periodic revivals to a low entropy state (blue fringes). Between revivals, the system has higher entropy (red fringes). The arrow of time points away from the low-entropy revivals and towards regions of higher entropy.
    (g, h) The system is initialized in a right-moving wavepacket centered at the origin. Near the packet's core, the arrow of time vectorfield exhibits combined temporal and spatial causal flow while remaining orthogonal to the quasi-particle's tangent vector. At larger distances from the center, the vectors become purely spacelike and continue to point toward the wavepacket’s center.
    \label{fig:aotexamples}}
\end{figure*}
%--------------------------------------

\section{Phenomenology}

We begin by performing a numerical study of the arrow of time vectorfield $\vv{\textsf{AOT}}_{\!t,x}$ in 1+1 local Hamiltonian dynamics for 20 to 100 qubits.  Later we will provide some analytic tools for more precisely interpreting our numerical findings.

\textbf{Example 1.}~First we consider a product state $|\phi\rangle$ and evolve it with a local Hamiltonian $H$ (in our numerics we pick the quantum Ising model).  Defining $|\Psi(t)\rangle \equiv e^{- i H t}|\phi\rangle$, the state $|\Psi(0)\rangle = |\phi\rangle$ has zero local entanglement entropy, and both $|\Psi(-T/2)\rangle$ and $|\Psi(T/2)\rangle$ have non-zero local entanglement entropy for $T > 0$ (see \cref{fig:aotexamples}(a) for a depiction).  Computing the $\vv{\textsf{AOT}}_{\!t,x}$ vectorfield for these dynamics in \cref{fig:aotexamples}(b), we see that the arrow of time points towards the future for $0 \leq t \leq T/2$ and towards the past for $-T/2 \leq t \leq 0$.  This is a quantum manifestation of our thermodynamic intuition that the flow of causal influence should be related to the local flow of entropy.
\vspace*{.2cm}

\textbf{Example 2.}~Building on the previous example, we construct an example of the evolution of a pure state in which there are two arrows of time side-by-side, one pointing towards the future and the other pointing towards the past.  See \cref{fig:aotexamples}(c) and (d). In the $\vv{\textsf{AOT}}_{\!t,x}$ vectorfield, we see that on the left side of the plot the arrow of time runs from smaller $t$ to larger $t$, whereas on the right side of the plot the arrow of time runs in the opposite direction.  There is an interface between the two regions, where the arrow of time `switches' from being future-facing to being past-facing. At this interface, the spatial component of $\aotv$ always points across the boundary; its exact magnitude is determined by a specific set of Hamiltonian correlators we discuss later.  The interface spreads out at timescales longer than depicted in the Figure.
\vspace*{.2cm}

\textbf{Example 3.}~Now we turn to the PXP model with quantum scarring, defined by the Hamiltonian $H_{\text{PXP}} = \sum_i P_{i-1}\,X_i\,P_{i+1}$ with projectors $P_i = (\mathds{1} + Z_i)/2$, and initialized in the Néel state $\ket{\Psi} = \ket{0101\cdots01}$~\cite{bernien2017probing, turner2018weak, turner2018quantum}. Under these conditions the system exhibits periodic revivals, oscillating between high local entropy and low local entropy states, as is a hallmark of quantum many-body scars. The temporal component of the $\aotv$ vectorfield undergoes a periodic inversion that coincides with the entropy oscillations. This behavior illustrates the arrow of time's application even to non-ergodic dynamics: even though systems with scarring resist typical thermalizing behavior, the arrow of time nonetheless follows the local entropy gradient.
\vspace*{.2cm}

\textbf{Example 4.}~Our last example consists of quantum Ising model dynamics in which a (non-topological) quasi-particle propagates from left to right~\cite{milstedCollisionsFalseVacuumBubble2022, vandammeRealtimeScatteringInteracting2021}.  This is illustrated in \cref{fig:aotexamples}(g) and (h). The subsystem entropy in this example is simple to understand. It is large (red) around the center of the quasi-particle's wavepacket, where the quasi-particle is mostly concentrated. Away from this region, the lattice has a well-defined, low-entropy spin configuration (blue). The $\aotv$ vectorfield has two interesting properties: (i) Near the center of the wavepacket it points directly towards the quasi-particle's worldline, which means that in the high-entropy core the causal influence flows both forward in time and inward in space onto the particle’s trajectory.  (ii) Farther from the center its temporal component vanishes so there is no net forward-in-time arrow while its spatial component remains nonzero and continues to point towards the same worldline, indicating a purely spatial ``pull'' toward the quasi-particle even in low-entropy regions. Finally, note that under a formal time-reversal (i.e.~$t\mapsto -t$ and reversing the quasi-particle's momentum), the $\aotv$ vectorfield remains invariant.
\vspace*{.2cm}
%--------------------------------------

%--------------------------------------

\section{Short-time expansion}

Quantum entanglement enters $\aotv$ in two distinct ways: (i) the entanglement already present in $\ket{\Psi(t)}$, and (ii) the entanglement that builds over the evolution $U(\Delta t) \equiv \exp(-iH\Delta t)$. When $H$ is a local Hamiltonian on $d$-dimensional lattice and $\Delta t$ is small, these two contributions can be separated precisely. In that regime, the pre-existing entanglement provides the leading contribution to the arrow of time, while the entanglement built up during the short evolution gives a subleading correction. These leading and subleading terms make up the temporal and spatial components of $\aotv$, respectively.

In the small $\Delta t$ limit, the leading approximation of $\aotv$ is purely temporal, and its magnitude scales with the change in local purity at $x$ over the interval $\Delta t$. A series expansion of $\aotv$ around $\Delta t = 0$ gives the leading order expression
\newpage
\begin{align}
    \aotv &\simeq \frac{4i\Delta t}{d(d^{2}+1)}\imgbox{50pt}{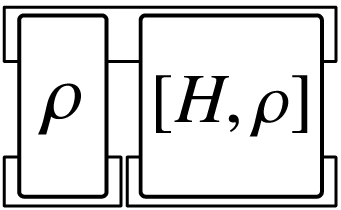}\,\vec{v}_{(t,x),(t+\Delta t, x)}\label{eq:aotEntApprox} \\
    &\simeq
     \frac{2(e^{-S_2(\rho_{\!x}(t))} - e^{-S_2(\rho_{\!x}(t+\Delta t))})}{d(d^{2}+1)} \,\vec{v}_{(t,x),(t+\Delta t, x)} \nonumber
\end{align}
where $S_2(\rho_x) \equiv -\log(\tr(\rho_x^2))$ is the second-order Rényi entropy. This result shows that for a local Hamiltonian, the leading order contribution to the arrow of time vector points from regions of lower local entropy toward regions of higher local entropy. Such behavior is readily seen in our numerics, particularly in \cref{fig:aotexamples}(b) and (f), which have temporal arrow of time vectors aligned with the local entropy gradient.

The arrow of time vectorfield intuitively represents flow from low to high entropy, reflecting the familiar thermodynamic arrow of time. Yet $\aotv$ departs from this purely entropic picture by developing spatial components that do not follow the local entropy gradient.  Our numerics show clear instances of this in \cref{fig:aotexamples}(d) and (h). In both cases, the Hamiltonian is identical to that of \cref{fig:aotexamples}(b), which has a purely temporal arrow of time. Therefore, the spatial components of \cref{fig:aotexamples}(d) and (h) are a consequence of their different initial states $\ket{\Psi}$ and subsequent evolution $\ket{\Psi(t)}$. Below we discuss the precise mechanism that determines the spatial components from the state and Hamiltonian.

The subleading approximation of $\aotv$ is of order $\cO(\Delta t^2)$ and manifests as a purely spatial component. This subleading term, together with higher-order corrections, dictates how the arrow of time veers away from the local entropy gradient. To track this deviation, we analyze the individual causal influences that feed the spatial component of $\aotv$, concentrating on the four off-diagonal terms $C_1$, $C_3$, $C_5$, and $C_7$, highlighted in Fig.~\ref{fig:spacetimelattice}. These four influences quantify how a disturbance located one step away in both space and time, for example at $(t+\Delta t,x\pm\Delta x)$ or $(t-\Delta t,x\pm\Delta x)$, reaches the point $(t,x)$. Because such diagonal propagation is suppressed in product states and activated only by entanglement-driven operator spreading, the relative sizes of $C_1$, $C_3$, $C_5$, and $C_7$ set the spatial tilt of $\aotv$: their vector sum fixes not just the magnitude but the precise sideways direction in which the arrow of time deviates from the local entropy gradient.

\begin{figure}
    \centering
    \includegraphics[]{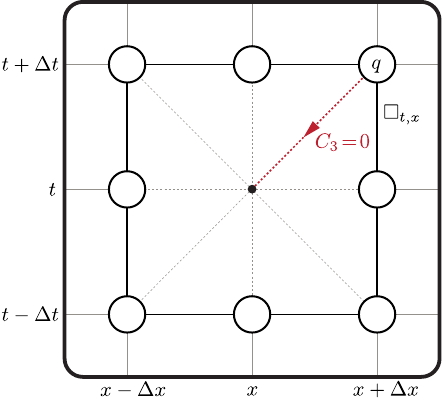}
    \caption{\textbf{Spacetime lattice illustrating \cref{thm:acausal}.} The lattice shows the neighborhood $\square_{t,x}$ around the point $(t,x)$ (black dot), with neighbor $q = (t+\Delta t, x+\Delta x)$. The red dashed arrow indicates a causal influence $C_3 = \overline{\text{CI}}_{q,(t,x)}$ that vanishes if and only if the conditions of \cref{thm:acausal} are satisfied.}
    \label{fig:theorem}
\end{figure}

\section{Characterization of local acausality}

The arrow of time at a given point depends on the causal influences exerted on the point by its neighbors, as reflected in \eqref{eq:MAIN_aot_def}. By selectively setting some of these causal influences to zero, we can control the direction of $\aotv$. In the following theorem, we give necessary and sufficient conditions for a state $\ket{\Psi(t)}$ to have a given causal influence vanish to all orders in $\Delta t$.

\begin{theorem}\label{thm:acausal}
Let $(t,x)$ be a lattice point and  
$q\in\square_{(t,x)}$ a neighboring site containing one qubit.
Set $(\tau,\xi)=q-(t,x)$ and write the state at time~$t$ in Schmidt form
\begin{align}\label{eq:thm_schmidt}
|\Psi(t)\rangle\;=\;\sqrt{p_1}\,|\psi_1\rangle_q|\phi_1\rangle_{q^c} +\sqrt{p_2}\,|\psi_2\rangle_q|\phi_2\rangle_{q^c}
\end{align}
where $q^c$ is the complement of $q$. For each $\alpha\in\{X,Y,Z\}$ expand the Heisenberg-evolved Pauli at site $x$ as
\begin{align}
\label{eq:thm_state_decomp}
\sigma_x^{\alpha}(\tau)
= \idm_q \otimes \nu^{\alpha0}_{q^c}(\tau) +\sum_{\beta=X,Y,Z}\sigma_q^{\beta} \otimes \nu^{\alpha\beta}_{q^c}(\tau).
\end{align}
Then the causal influence from $q$ to $(t,x)$ vanishes, $\overline{\mathrm{CI}}_{q,(t,x)}=0$, if and only if for every $\alpha,\beta\in\{X,Y,Z\}$
\begin{align}
\label{eq:thm_conds}
p_1\langle\phi_1|\nu^{\alpha\beta}_{q^c}|\phi_1\rangle =p_2\langle\phi_2|\nu^{\alpha\beta}_{q^c}|\phi_2\rangle
\,\, \text{and}\,\,
\langle\phi_1|\nu^{\alpha\beta}_{q^c}|\phi_2\rangle&=0\,.
\end{align}
\end{theorem}
\noindent In words, the two Schmidt sectors must give the same diagonal expectation value of every $\nu^{\alpha\beta}$, and all corresponding off-diagonal matrix elements must vanish.  If either condition fails for any $(\alpha,\beta)$, then $\overline{\mathrm{CI}}_{q,(t,x)}>0$. We provide a proof of the above theorem, as well as statements regarding the precise value of the causal influence, in the Appendix.

An instructive application of Theorem~\ref{thm:acausal} is to see explicitly how quantum entanglement can nullify certain causal influences.  Let us focus on the nearest-neighbor Ising Hamiltonian $H = \sum_{k=1}^{N} X_k X_{k+1}$; take $(t,x)$ as in the Theorem, let $q$ be its neighbor displaced by $(\tau,\xi)$, and define $x' \equiv x-\xi$.  A direct Heisenberg-picture calculation gives the operators defined in \cref{thm:acausal}:
\begin{align}
    \nu^{YX}_{q^c}(\tau) &= \cos(2\tau)\,
        Z_x \otimes \idm_{\{x,q\}^c}
        \\
        & \quad\,+
        \sin(2\tau)\,
        X_{x'} \!\otimes\! Y_x \otimes \idm_{\{x',x,q\}^c}\,, \nonumber \\
    \nu^{ZX}_{q^c}(\tau) &= \cos(2\tau)\,
        Y_x \otimes \idm_{\{x,q\}^c} \\
        & \quad \, -
        \sin(2\tau)\,
        X_{x'} \!\otimes\! Z_x \otimes \idm_{\{x',x,q\}^c}\,. \nonumber 
\end{align}
Define two orthonormal states on $q^c$,
\begin{align}
    |\phi_1\rangle_{q^c} &=
        |+\rangle_{x'} \otimes
        \bigl(e^{+i X\tau}|0\rangle_x\bigr) \otimes
        |0\cdots0\rangle_{\{x',x,q\}^c},\\
    |\phi_2\rangle_{q^c} &=
        |-\rangle_{x'} \otimes
        \bigl(e^{-i X\tau}|0\rangle_x\bigr) \otimes
        |0\cdots0\rangle_{\{x',x,q\}^c},
\end{align}
and set $p_1=p_2=\tfrac12$.  With these choices the conditions~\eqref{eq:thm_conds} are met, so the purified state $\ket{\Psi(t)} = \sqrt{p_1}\,\ket{\psi_1}_{q}\ket{\phi_1}_{q^c} + \sqrt{p_2}\,\ket{\psi_2}_{q}\ket{\phi_2}_{q^c}$ obeys $\overline{\mathrm{CI}}_{q,(t,x)} = 0$: the qubit at $q$ cannot causally influence the site $(t,x)$.

Operationally, one prepares $\ket{\Psi(t)}$ by placing $q$ in the superposition $\tfrac{1}{\sqrt2}(\ket{0}+\ket{1})$, rotating site $x$ about the $X$–axis by $\pm\tau$ conditioned on the state of $q$, and recording the rotation's sign in the $X$-eigenstate $\ket{\pm}$ of site $x'$.  The vanishing causal influence relies on the entanglement between $q$ and $\{x',x\}$; removing either branch of the superposition restores a non-zero $\overline{\mathrm{CI}}_{q,(t,x)}$,  showing that the suppression is a genuinely quantum (and not merely classical) effect.

\section{Application to quantum error-correction}

When quantum many-body systems are implemented through error-corrected quantum simulation~\cite{chenFaultTolerantQuantum2025, daleyPracticalQuantumAdvantage2022, whitfieldSimulationElectronicStructure2011}, their causal structure becomes significantly more complex and has been largely unexplored.
This raises the question: when a many-body system is embedded in an error-correcting code, how does the code's recovery cycle reshape its causal structure -- and, in particular, does a qubit that experiences a fault retain any causal memory of its past after correction?

To address this, we incorporate active error correction into the causal influence framework by modifying Eq.~\eqref{eq:cinfAB} to account for a recovery map. We assume the recovery map is local, mapping local regions to other local regions. We substitute the physical region $B$ with the logical region $\overline B$ and replace the operator $O_B$ with $\mathcal R^{\dag}\!\left[O_{\overline B}\right]$, where $O_{\overline B}$ is any Hermitian operator acting on the logical qubits and $\mathcal R$ is the (CPTP) recovery map of the stabilizer code; $\mathcal R^{\dag}$ denotes its Hilbert--Schmidt adjoint. For such a code, the recovery channel can be written as $\mathcal R[\rho]=\sum_{s} C_s \Pi_s\rho\,\Pi_s C_s^{\dag}$, with $\Pi_s$ projecting onto the syndrome-$s$ subspace and $C_s$ applying the corresponding correction. These substitutions leave Theorem~\ref{thm:acausal} untouched, so exactly the same state-dependent conditions decide when the \emph{error-corrected causal influence} (ECI) $\overline{\mathrm{CI}}_{A\overline B}$ vanishes.

As a warmup, we begin by studying how encoding alone affects causal influence, using the Iceberg code~\cite{selfProtectingExpressiveCircuits2024}. This stabilizer code encodes an even number $k$ of logical qubits into $n=k+2$ physical qubits labeled $\{1,\dots ,k,h,v\}$. The code is stabilized by $X^{\otimes n}$ and $Z^{\otimes n}$; single-qubit logical operators are $\overline X_j = X_j X_h$ and $\overline Z_j = Z_j Z_v$ for $j = 1,\dots ,k$. For the simplest interacting Hamiltonian that acts identically at logical and physical levels,
\begin{align}
H_0 = \sum_{j=1}^{k-1} \overline X_j \overline X_{j+1}
    = \sum_{j=1}^{n-3} X_j X_{j+1}\,,
\end{align}
we find that, starting from an encoded basis state $\ket{\Psi}=\ket{\overline b}$, the causal influence of any physical qubit on itself vanishes to all orders in $\Delta t$. Thus a physical qubit cannot affect its own future even though it can mediate correlations with distant qubits, consistent with error-protected dynamics where local perturbations cannot persist independently because of the code's global correlations.  Adding a logical transverse field breaks this perfect cancellation: $H_L = H_0 + h_Z\sum_{j=1}^k \overline{Z}_j$. However, when implementing the evolution through a Trotter decomposition, careful local frame changes (e.g., Hadamard layers swapping $X\!\leftrightarrow\! Z$) allow one to steer the causal flow through designated circuit elements, preserving the suppression of self-influence outside those layers.

\begin{figure*}
    \centering
    \includegraphics[width=1\linewidth]{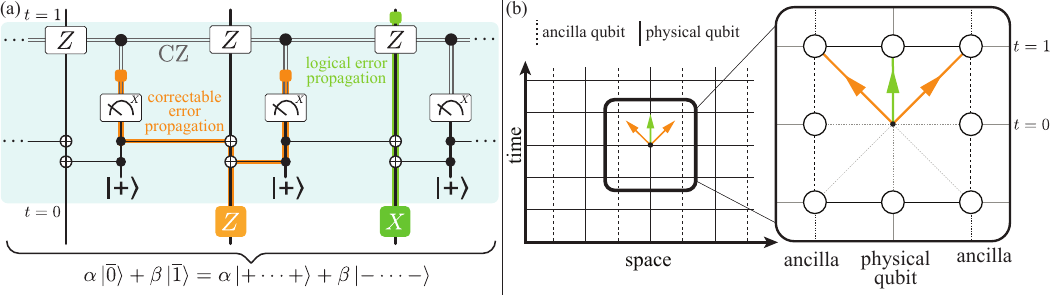}
    \caption{(a) One error-correction cycle in the 1D $X$-basis repetition code correcting single-qubit $Z$ errors. The initial state is the logical superposition $\alpha\ket{\overline{0}} + \beta\ket{\overline{1}} = \alpha\ket{+\cdots+} + \beta\ket{-\cdots-}$. Time runs from $t=0$ to $t=1$; the blue band highlights one QEC cycle. Alternating physical and ancilla qubits form two-site unit cells. Each ancilla is prepared in $|+\rangle$, coupled by $\mathsf{CX}$ gates to its two neighboring physical qubits to extract an $X$-parity syndrome, and then measured in the $X$ basis. The orange and green lines indicate pathways along which physical $Z$ and $X$ errors propagate under time evolution. Two operator species are shown: a physical $X$ (i.e., a logical $\overline Z$) commutes through the circuit and generates a \emph{strictly temporal} arrow of time (it propagates upward, leaving no syndrome record), whereas a physical $Z$ error is detected and removed so that its arrow of time \emph{terminates} at the syndrome measurement. (b) Visualization of two types of causal influences on the spacetime lattice. The green arrow represents the influence exerted by a physical $X$ perturbation (i.e., a logical $\overline Z$ error) at $t = 0$ on the same physical qubit at $t = 1$. The orange arrows represent the influence exerted by the correctable $Z$ perturbation at $t = 0$ on the ancillas at $t = 1$.}
    \label{fig:qec_cycle}
\end{figure*}

We now incorporate active error correction and begin with the simplest example that exposes the key features. Figure~\ref{fig:qec_cycle} shows one cycle of the 1D $X$-basis repetition code, in which alternating physical and ancilla qubits form two-site unit cells. The logical basis states in this code are $\ket{\overline0} = \ket{+\cdots+}$ and $\ket{\overline1} = \ket{-\cdots-}$. Each ancilla is prepared in $\ket{+}$, interacts with its two neighboring physical qubits via $\mathsf{CX}$ to extract an $X$-parity syndrome, and is then measured in the $X$ basis; the measurement outcome drives a (Pauli-frame) $Z$ correction on the physical qubits. This circuit induces qualitatively different causal flows for different operator types. For example, a physical $X$ error on any qubit acts as a logical $\overline Z$ operator on the state $\alpha\ket{+++} + \beta\ket{---}$, transforming it to $\alpha\ket{+++} - \beta\ket{---}$. Such an error goes undetected by the code, and therefore no causal influence is exerted on the ancilla qubits adjacent to the erred physical qubit. In general, a logical perturbation induces a causal influence that is purely in the temporal direction: it flows straight forward to later times without depositing information in the ancillas. By contrast, a single-qubit $Z$ error is \emph{correctable}, transforming the state to, e.g.~$\alpha\ket{+-+} + \beta\ket{-+-}$. The error is identified by the $X$-parity syndrome measurement and subsequently corrected. Thus, its causal influence to the same physical site vanishes after correction, i.e.~its arrow of time terminates at syndrome extraction. This sharp distinction between logical and correctable operators is what appears directly in the causal influence structure of the code.

Building on this intuition, let us analyze the $[[5n,n,3]]$ stabilizer code constructed from $n$ copies of the perfect $[[5,1,3]]$ code. Under the logical Hamiltonian
\begin{align}
H_0 = \sum_{j=1}^{n-1} \overline X_j \,\overline X_{j+1},
\end{align}
we can apply Theorem~\ref{thm:acausal} to find states with vanishing error-corrected causal influence. The analysis yields operators
\begin{align}
\begin{split}
	\nu^{Y\!X}_{q^c}\!\!&=\! -\tfrac{1}{2}\!\sin(4 \Delta t)\, Y_{\! 5j-2}Y_{\! 5j-1}\!\!\left[\tan(2 \Delta t)\,\overline{Y}_{\! j}\,\overline{X}_{\! j+1} \!+\! \overline{Z}_{\! j}\right]\!,\\
	\nu^{Z\!X}_{q^c} \!\!&=\! -\tfrac{1}{2}\!\sin(4 \Delta t)\,  Y_{\! 5j-2}Y_{\! 5j-1}\!\!\left[\tan(2 \Delta t)\,\overline{Z}_{\! j}\,\overline{X}_{\! j+1} \!-\! \overline{Y}_{\!j}\right]\!.
\end{split}
\end{align}
States satisfying the constraints from these operators reveal a class of intrinsically protected logical states that remain isolated from physical errors even when active error correction is imperfect. These include product states where logical qubits are in $\overline X$-eigenstates, as well as genuinely entangled choices like the Bell-like superposition $\frac{1}{\sqrt{2}}(\ket{\overline{00}}+\ket{\overline{11}})$.

For a more complete picture, let us analyze causal influences during a single error-correction cycle for general stabilizer codes. Instead of the recovery channel $\mathcal{R}$ alone, we use the dilated channel $\mathcal{R}'[\rho] = \sum_{s,t} C_s \Pi_s \rho \Pi_t C_t^\dag \otimes \ket{s}\bra{t}_{\mathrm{anc}}$ that describes the joint evolution of physical and ancilla qubits (with $\mathcal R=\mathrm{tr}_{\mathrm{anc}}\mathcal R'$). In ideal single-error-correcting stabilizer codes, two influences vanish: logical $\rightarrow$ ancilla (syndromes reveal no information about logical operations) and physical $\rightarrow$ logical (single-site perturbations are fully removed). The influence that survives is logical$\rightarrow$logical, which remains non-zero with magnitude $\overline{\mathrm{CI}}_{LL}=(D-1)/[D^2(D^2+1)]$ for codespace dimension $D=2^k$. Meanwhile, the physical $\rightarrow$ ancilla influence quantifies information flow from potential error locations to the syndrome record, serving as the diagnostic channel for error detection.

Our findings in this section reveal that causal influence analysis can pinpoint logical qubits that are intrinsically protected from specific error sources, offering a valuable complement to conventional fault tolerance. Although the practical gains must still be quantified, integrating these diagnostics could help target active correction to the causally exposed parts of a simulator while allowing shielded subsystems to evolve with smaller overhead.

%--------------------------------------

\begin{figure}
    \centering
    \includegraphics[]{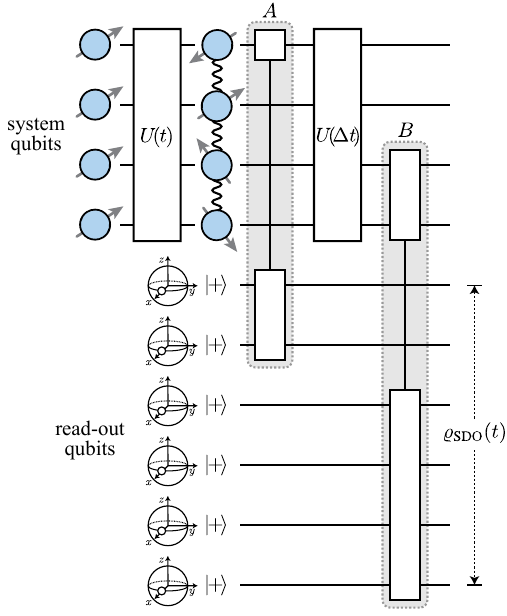}
    \caption{The quantum circuit that generates the quantum superdensity operator $\varrho_{\text{SDO}}(t)$. The measurement procedure accepts as input a quantum system in an initial state $\ket{\Psi}$.}
    \label{fig:measurementCircuit}
\end{figure}

\section{Measuring the arrow of time vectorfield}

While the quantum many-body systems we have considered here are simple enough to classically simulate, understanding causal influence and the arrow of time in the presence of strong interactions would be valuable to explore. As such, we show how causal influence quantities can be efficiently extracted in state-of-the-art quantum simulators. 

The key intuition is that the superdensity operator~\cite{cotlerSuperdensityOperatorsSpacetime2018} can be directly accessed in a quantum simulation experiment by using ancilla qubits (see~\cref{fig:measurementCircuit}). This is done by mapping the target superdensity operator to the density matrix of an ancilla register that has been entangled with the primary physical system across multiple time steps. Elsewhere in literature, this is sometimes understood as a ``space-to-time duality/mapping", and is how one can experimentally probe spatiotemporal correlations in an evolving system.

To be concrete, consider a region $A$ with $n$ qubits and a region $B$ with $m$ qubits. Further, for simplicity, assume $A$ and $B$ are entirely supported at $t=t_A$ and $t=t_B$ respectively. In general, our procedure only requires that each region is space-like (i.e.~the causal influence between any two sites in $A$ is zero, and respectively for $B$). 

The next step is to bring in an ancilla register and couple it to the primary system at times $t_A$ and $t_B$. Having $2(n+m)$ qubits is sufficient to construct the superdensity operator. One starts by preparing all ancilla qubits in the $\vert + \rangle = \frac{1}{\sqrt{2}}\left( \vert 0 \rangle + \vert 1 \rangle \right)$ state, and evolves the system up until $t_A$. Then, apply $\textsf{CX}$ gates transversally between the first $n$ ancilla qubits and region $A$. Next, one applies $\textsf{CZ}$ gates between the next $n$ ancilla qubits and region $A$. In both cases, the ancillas act as the control qubits for the $\textsf{CX}$ and $\textsf{CZ}$ gates. Measurements of the first $2n$ ancilla qubits probe region $A$ in both $X$ and $Z$ quadratures, providing full access to all local density matrices in $A$. After the ancilla/system interaction, the expectation value of any operator in the region can be estimated by measuring the ancilla register.

For the final step, one evolves the system to time $t_B$. This may involve applying a backwards time evolution, which can be implemented by flipping the sign of $H$. The procedure performed for region $A$ is repeated for region $B$, applying the controlled gates $\textsf{CX}$ and $\textsf{CZ}$ between the final $2m$ ancilla qubits and $B$. After this second interaction, the ancilla register contains all two-time correlation functions $\langle O_B(t_B) O_A(t_A) \rangle$ supported on the regions $A$ and $B$. As such, the ancilla register stores a complete representation of the superdensity operator
\begin{align}
\varrho_{\mathrm{SDO}} = \sum_{O_A, O_B} O_A \otimes O_B \langle O_B(t_B) O_A(t_A)\rangle.
\end{align}
Properties of the superdensity operator~\cite{cotlerSuperdensityOperatorsSpacetime2018} can subsequently be estimated by taking measurements of the ancillas. The causal influence is then extracted from a set of R\'{e}nyi-2 entropies collected with respect to certain partitions of the ancilla register, the precise details of which are given in the Appendix.  The R\'{e}nyi-2 entropies can be directly estimated by performing measurements in the Bell basis across two copies of the prepared state~\cite{bluvsteinLogicalQuantumProcessor2024}.

%--------------------------------------

\section{Discussion and outlook}

We have developed new tools to characterize local arrows of time in quantum many-body systems, and explored different manifestations of local arrows of time in a variety of examples.  There are many other interesting examples which are beyond our ability to study using classical simulation. Thus contemporary quantum simulators could provide novel opportunities for measuring local arrows of times in settings beyond the purview of classical simulation.

These results expose a  link between microscopic ingredients, the Hamiltonian $H$ and the many-body state $|\Psi\rangle$, and the emergent arrow of time.  Whereas $H$ and $|\Psi\rangle$ have so far been inputs for computing the causal influence vectorfield, the vectorfield itself can also act as a diagnostic.  Tracking its short-time evolution provides a tomographic diagnostic that indicates which interaction terms amplify or quench local influence and, in parallel, displays how structural features of $|\Psi\rangle$ are imprinted in spacetime.  Hence measurements of causal influence could ultimately constrain, and even help identify, novel features of both the effective Hamiltonian and the underlying quantum state governing complex many-body dynamics.

Our examples also uncover an interesting connection between local arrows of time and quantum error correction. When a many-body system is encoded in a stabilizer code the code’s global constraints redistribute causal influence: individual physical qubits lose the ability to influence their own futures, while information is funneled through the shared logical degrees of freedom. In other words, the code ``straightens'' the local arrow of time by quenching microscopic self-influence and forcing causal flow to propagate non-locally across the code block. This suppression is invisible in thermodynamic diagnostics but becomes obvious once one tracks the full causal influence vectorfield. We have already analyzed ECI within a single QEC cycle for ideal stabilizer circuits, distinguishing logical components that propagate strictly forward in time from correctable faults that terminate at syndrome extraction. The next steps would be to extend this analysis to multiple cycles under realistic noise and imperfect recovery; to codes beyond our toy examples (e.g.~surface and LDPC codes); and to error models including coherent and leakage faults, while quantitatively correlating ECI flow with decoder performance and scheduling choices. Quantifying these patterns would clarify how information is redistributed among physical qubits, ancillas, and the recovered logical state, and could furnish operational diagnostics for decoder performance, leakage, and circuit-scheduling choices. Furthermore, applying this causal framework to recent spacetime quantum error-correction codes~\cite{gottesmanOpportunitiesChallengesFaulttolerant2022} could reveal new organizing principles. Indeed, quantum causal structure itself might serve as a design principle for developing novel error-correcting codes with tailored information flow properties. Finally, we defined an ``error-corrected causal influence'' for which our acausality theorem still holds, allowing us to follow how faults spread at the physical level and to pinpoint logical subsystems that remain shielded. In future work it would be valuable to apply this framework to a wider class of codes in order to map out systematic links between a code’s causal geometry and key performance metrics such as logical error rates, threshold values, and the emergence (or absence) of well-behaved local arrows of time.

%--------------------------------------

\emph{Acknowledgments.}
We would like to thank Pablo Bonilla, Sasha Geim, and Frank Wilczek for valuable discussions. We acknowledge financial support from the U.S. Department of Energy (DOE Quantum Systems Accelerator Center, grant number DE-AC02-05CH11231),
DARPA  MeasQuIT program (grant number HR0011-24-9-0359), the Center for Ultracold Atoms (an NSF Physics Frontier Center), the National Science Foundation (grant numbers PHY-2012023, CCF-2313084). 

\pagebreak
\onecolumngrid

\appendix
\appendixonlytocsetup
\printappendixtoc

\section{Notation and mathematical preliminaries}\label{sec:appNotation}
In this section we will introduce the notation and mathematical background used throughout the Appendix. The specialized notation will be useful for developing intuition about the causal influence, and for proving its various properties.

\subsection{Operators and projection superoperators}

We begin by defining $\cH$ as the vector space of $N$ qubits with the dimension $d = \dim\cH = 2^N$, and $\mathrm{End}(\mathcal{H})$ the operator space on these $N$ qubits. The standard basis of the operator space is the set of $N$-qubit Pauli strings, $\{\sigma^{i_1}_1\sigma^{i_2}_2 \cdots \sigma^{i_N}_N \,:\, i_1, \dots, i_N \in \{\idm, X, Y, Z\}\}$. Any operator $O\in\mathrm{End}(\mathcal{H})$ may be decomposed as
\begin{align}
	O &= \sum_{i_1, \dots, i_N \in \{\idm, X, Y, Z\}} O_{i_1\cdots i_N} \sigma^{i_1}_1\sigma^{i_2}_2 \cdots \sigma^{i_N}_N
\end{align}
where $\sigma^{i_k}_k$ is the operator $\idm$, $X$, $Y$, or $Z$ on the $k$th qubit. We will sometimes use $i_k = 0$ to denote $i_k = \mathds{1}$.  If we write $\sigma^i_S$ where $S$ is a subsystem of more than one qubit, then it refers to the Pauli string operator with the index $i$ acting on that subsystem.

The next step is to endow the operator space with an inner product, which will also motivate a braket-like notation~\cite{chenSpeedLimitsLocality2023} for operators that becomes useful later in the appendix. For two operators $A, B \in \mathrm{End}(\cH)$, we define their normalized Hilbert-Schmidt inner product as
\begin{align}
	(A|B) &\equiv \frac{\tr(A^\dag B)}{d}.
\end{align}
This naturally leads us to the vectorized forms of the operators $B$ and $A^\dag$, which we denote as $|B)$ and $(A|$, respectively. In this notation, a superoperator $\cT$, defined as a linear map on operators, acts as $|A) \mapsto \cT|A) \equiv |\cT(A))$.

The \textit{projection superoperator} $\bP_k$ acts as
\begin{align}
	\bP_k(\sigma^{i_1}_1 \cdots \sigma^{i_N}_N) = (1 - \delta_{i_k 0})\sigma^{i_1}_1 \cdots \sigma^{i_N}_N,
\end{align}
where $k \in \{1, 2, \dots, N\}$. We also define the projector for a subset of sites $S \subseteq \{1, 2, \dots, N\}$,
\begin{align}\label{eq:proj_op_def}
	\bP_S(\sigma^{i_1}_1 \cdots \sigma^{i_N}_N) &\equiv \left(1 - \prod_{k\in S}\delta_{i_k 0}\right)\sigma^{i_1}_1 \cdots \sigma^{i_N}_N.
\end{align}
In other words, $\bP_S$ projects out the part of an operator that acts trivially on $S$.

\subsection{Averaging with the Haar and Hilbert-Schmidt measures}

\noindent\paragraph*{Haar averages.} Expectation values over unitary operators $U\in\mathrm{End}(\cH)$ are taken with respect to the Haar measure $\mu(U)$:
\begin{align}
	\ex_{U}\left[f(U)\right] = \int d\mu(U)\, f(U).
\end{align}
For any operators $\rho,\sigma\in\mathrm{End}(\cH)$ one has
\begin{subequations}\label{eq:haar_ids}
\begin{align}
	\ex_{U}\left[U\sigma U^{\dag}\right] &= \frac{\tr(\sigma)}{d}\idm,\\
	\ex_{U}\left[ U^{\otimes 2}(\rho\otimes\sigma) U^{\dag\otimes 2}\right] &= \frac{1}{d(d^2 - 1)}\left[\tr(\rho)\tr(\sigma)\left(d\,\idm - \mathbb{F}\right) + \tr(\rho\sigma)\left(d\,\mathbb{F} - \idm\right) \right],
\end{align}
\end{subequations}
where $\mathbb{F}$ swaps the two tensor factors. The Haar variance is
\begin{align}
	\underset{U}{\mathrm{Var}}\left[f(U)\right] &\equiv \ex_U\left[f(U)^2\right] - \ex_U\left[f(U)\right]^2.
\end{align}

\paragraph*{Hilbert-Schmidt averages.} In the definition of the causal influence \eqref{eq:cinfAB} we must average over positive operators drawn from the Hilbert-Schmidt measure. The simplest way to define this is to write
$$
O = Q^{\dagger}Q,\qquad \text{with } \operatorname{tr}(Q^{\dagger}Q)=1 .
$$
To sample from the Hilbert-Schmidt measure, we draw $Q$ uniformly from the unit sphere defined by $\tr(Q^{\dagger}Q)=1$, and then form $O$. We define expectation values over $O$ accordingly,
\begin{align}
    \bE_{O} f(O) \equiv \int_{\tr(Q^\dag Q) = 1} dQ\, f(Q^\dag Q)
\end{align}
An important property of this measure is that if $f(Q) \geq 0$ for all $Q$ and $\ex_{Q} f(Q) = 0$, then $f(Q) = 0$ everywhere, since $f$ is continuous and the measure has full support.

\section{Approximations of the causal influence}
In the following, we provide the first few nontrivial terms of the causal influence \eqref{eq:cinfAB} as a power series in $\Delta t$. In the first part of the discussion, subsystem $A$ is defined at time $t$ and subsystem $B$ is defined at time $t + \Delta t$. This convention differs from that used in the main body of the paper, where the times for subsystems $A$ and $B$ are swapped from what they are here. We assume the system is in a pure state $\rho$, and give the causal influence of $A$ on $B$ in two separate cases.

The ``same-site'' causal influence is given by considering the same subsystem at different times, i.e., $A = B$. We have
\begin{align}
\begin{split}
    \overline{\mathrm{CI}}_{AA}
     &\simeq \frac{1}{d_A(d_A^{2}+1)}\Bigg\{\tr(\rho_A^2) - \frac{1}{d_A} -\frac{2\Delta t^2}{d_A^2-1}\bigg[\tr\big( \tr_{A^c}\left(H\rho_{A^c}\right)^2 - \frac{1}{2}\|\left[\tr_A(H),\rho_{A^c}\right]\|^2_2 - H^2\rho_{A^c}\big) \\
     &\qquad + \frac{d_A}{2}\|\left[H, \rho_{A^c}\right]\|^2_2 + \frac{1}{d_A}\Big(\tr\left(\tr_A(H)^2\rho_{A^c}\right) - \tr\left(H\rho_{A^c}\right)^2\Big) \bigg]\Bigg\}
\end{split}
\end{align}
where $\rho_S \equiv \idm_{S^c}\otimes\tr_{S^c}(\rho)$ for any subsystem $S$.

Alternatively, $A$ and $B$ do not overlap at all, $A\cap B = \emptyset$. Then,
\begin{align}\label{eq:cinfAB2}
    \overline{\mathrm{CI}}_{AB}
     \simeq \frac{2\Delta t^2}{d_B(d_A^2 - 1)(d_B^2 + 1)}\mathrm{Re}\,\tr\left\{ \rho_{A^c}V_{AA^c}\tr_{AB}(\rho)V_{AA^c} - \tr_B\left(\rho_{A^c}V_{AA^c}\right)^2 + \frac{1}{d_A^2}\tr_{(AB)^c}\!\left([V_{AA^c}, \rho_{A^c}]\right)^2 \right\}
\end{align}
where $V_{AA^c} \equiv H - \tr_{A^c}(H)\otimes\idm_{A^c}/d_{A^c} - \idm_A\otimes\tr_{A}(H)/d_A$
is the term in $H$ that interacts $A$ and $A^c$. We can uniquely decompose the interaction as
\begin{align}
    V_{AA^c} = V_{AB} + V_{AE} + V_{ABE},
\end{align}
where $E \equiv (A\cup B)^c$ and each term in the expression acts nontrivially on $A\otimes B$, $A\otimes E$, and $A\otimes B\otimes E$. If $V_{ABE} = 0$ (e.g., a nearest-neighbor Hamiltonian), then \eqref{eq:cinfAB2} simplifies to
\begin{align}\label{eq:diffsiteABEzero}
    \overline{\mathrm{CI}}_{AB}
     \simeq \frac{2\Delta t^2}{d_B(d_A^2 - 1)(d_B^2 + 1)}\left[\tr\left(\tr_A\left(\rho_{AB}^2\right)\tr_A\left(V_{AB}^2\right) - V_{A_1B}\rho_{A_2B}V_{A_1B}\rho_{A_2B}\right) - \frac{1}{2d_A} \|\left[V_{AB}, \rho_B\right]\|_2^2\right].
\end{align}

Restoring the arrow of time convention, where $A$ and $B$ are defined at times $t - \Delta t$ and $t$, respectively, we have
\begin{align}
    &\cinf_{AA} \simeq \frac{1}{d_A(d_A^2+1)}\left\{\imgbox{35pt}{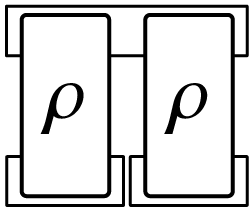} - \frac{1}{d_A} - 2i\Delta t\imgbox{50pt}{rhoA_Hrho}\right\},\\
    &\cinf_{AB}
     \simeq \frac{2\Delta t^2}{d_B(d_A^2 - 1)(d_B^2 + 1)}\left\{\imgbox{90pt}{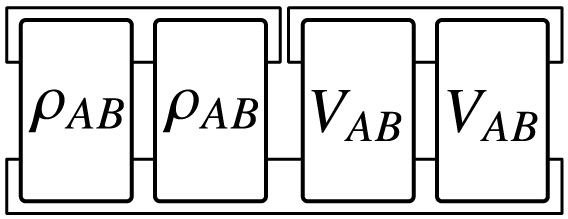} - \imgbox{110pt}{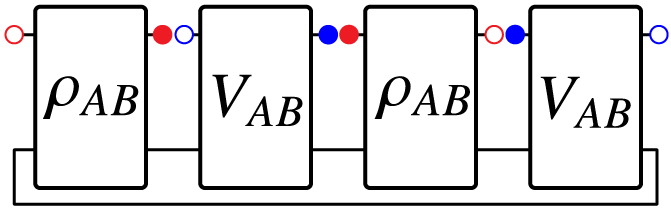} - \frac{1}{2d_A} \|\left[V_{AB}, \rho_B\right]\|_2^2\right\},
\end{align}
where higher orders in $\Delta t$ have been suppressed. These two equations correspond to \eqref{eq:cinfAB2} and \eqref{eq:diffsiteABEzero}, respectively.

\section{Alternate form of the causal influence}\label{sec:altFormCinf}

This section derives and interprets an alternate form of the causal influence. It is more amenable to interpretation because it separates the causal influence into state- and dynamics-dependent terms whose interaction determines the exact value of the causal influence. We clarify this point after the derivation.

\subsection{Derivation of the alternate form}\label{app:altFormDerive}

To begin, expand the definition \eqref{eq:cinfAB} of $\cinf_{AB}$,
\begin{align}
    \cinf_{AB} &= \bE_{O_B}\var_{V_A}\tr\left(U^\dag(t_A) V^\dag_A U(t_B - t_A)^\dag O_B U(t_B - t_A) V_A U(t_A)\rho\right).
\end{align}
Then, redefine the initial state by having it absorb the factors of $U(t_A)$ (i.e., $U(t_A)\rho U(t_A)^\dag \mapsto \rho$). We are left with
\begin{align}
    \cinf_{AB} &= \bE_{O_B}\var_{V_A} \tr\left(V^\dag_A O_B(t) V_A \rho\right),
\end{align}
where $t \equiv t_B - t_A$. In this section, we will focus on the inner part of the expression
\begin{align}
    \cinf_{A}(O_B) \equiv \var_{V_A} \tr\left(V^\dag_A O_B(t) V_A \rho\right),
\end{align}
where $\cinf_{A}(O_B)$ indicates a causal influence that explicitly depends on the observable $O_B$.

To derive the alternate form of the causal influence, start with the operator expansion
\begin{align}
    O_B(t) &= \frac{1}{\sqrt{d}}\sum_{i,j} c_{ij}(t) \,\sigma^i_A \sigma^j_{A^c}\,,\label{eq:operBAltDeriv}
\end{align}
and the identity
\begin{align}
    \underset{V_A}{\mathrm{Var}}\left[ \tr (V_A^{\dagger} O_B(t) V_A \rho) \right] &= \tr \left\{ \left( \ex_{V_A} \left[\left(V_A^{\dagger} O_B(t) V_A\right)^{\otimes 2} \right] - \ex_{V_A} \left[V_A^{\dagger} O_B(t) V_A\right]^{\otimes 2} \right) \rho^{\otimes 2} \right\}.\label{eq:altDeriveVar}
\end{align}
Then apply \eqref{eq:haar_ids} to obtain
\begin{align}
    \ex_{V_A} \big[ V_A^{\dagger \otimes 2} O_B(t)^{\otimes 2} V_A^{\otimes 2} \big] &= \frac{1}{d}\sum_{i,j,k,l} c_{ij}(t) c_{kl}(t) \ex_{V_A} \big[ V_A^{\dagger \otimes 2} (\sigma_A^i \otimes \sigma_A^k) V_A^{\otimes 2} \big] (\sigma_{A^c}^j \otimes \sigma_{A^c}^l) \\
    &= \frac{1}{d (d_A^2 -1)} \sum_{i,j,k,l} c_{ij}(t) c_{kl}(t) \big[ d_A \delta^{i0} \delta^{k0} (d_A \idm_A - \bF_A) + \delta^{ik} (d_A \bF_A - \idm_A) \big] (\sigma_{A^c}^j \otimes \sigma_{A^c}^l)
\end{align}
and
\begin{align}
    \ex_{V_A} \big[ V_A^{\dagger} O_B(t) V_A \big]^{\otimes 2} &= \frac{1}{d}\sum_{i,j,k,l} c_{ij}(t) c_{kl}(t) \big( \ex_{V_A} [V_A^{\dagger} \sigma_A^i V_A] \otimes \ex_{V_A} [V_A^{\dagger} \sigma_A^k V_A] \big) (\sigma_{A^c}^j \otimes \sigma_{A^c}^l) \\
    &= \frac{1}{d} \sum_{i,j,k,l} c_{ij}(t) c_{kl}(t) \delta^{i0} \delta^{k0} \idm_A (\sigma_{A^c}^j \otimes \sigma_{A^c}^l)
\end{align}
where $d_S \equiv \dim\cH_S$ for any subsystem $S$ and $\bF_A$ swaps the two $A$ subsystems in the doubled Hilbert space. Together, they imply
\begin{align}
	\ex_{V_A} \left[\left(V_A^{\dagger} O_B(t) V_A\right)^{\otimes 2} \right] - \ex_{V_A} \left[V_A^{\dagger} O_B(t) V_A\right]^{\otimes 2} &= \frac{1}{d_{A^c}}\sum_{i,j,k,l} c_{ij}(t) c_{kl}(t)  G^{ik}_A (\sigma_{A^c}^j \otimes \sigma_{A^c}^l)
\end{align}
where
\begin{align}
	G^{ik}_A &\equiv \frac{1}{d_A^2 - 1}\left(\delta^{ik} - \delta^{i0}\delta^{k0}\right)\left(\bF_A - \frac{\idm_A}{d_A}\right).
\end{align}
Note that $G^{i0}_A = G^{0k}_A = 0$, so that \eqref{eq:altDeriveVar} becomes
\begin{align}
	\underset{V_A}{\mathrm{Var}}\left[ \tr (V_A^{\dagger} O_B(t) V_A \rho) \right] &= \sum_{j,l}\underbrace{\frac{d}{d_A^2 - 1}\tr\left[\left(\bF_A - \frac{\idm_A}{d_A}\right) \frac{\sigma_{A^c}^j \otimes \sigma_{A^c}^l}{d_{A^c}} \rho^{\otimes 2}\right]}_{\textstyle\Theta_{jl}}\sum_{i} \frac{c_{ij}(t)^{*}c_{il}(t)}{d}(1 - \delta^{i0}).
\end{align}
Using \eqref{eq:operBAltDeriv} we obtain 
\begin{align}
	\frac{c_{ij}(t)^{*}c_{il}(t)}{d}(1 - \delta^{i0}) &= \frac{\tr(\sigma^i_A \sigma^j_{A^c}\bP_A O_B(t))^{*}\tr(\sigma^i_A \sigma^l_{A^c}\bP_A O_B(t))}{d^2} = \underbrace{(\bP_A O_B(t)|\sigma^i_A \sigma^j_{A^c})}_{\textstyle\Gamma_{ij}^{*}} \underbrace{(\sigma^i_A \sigma^l_{A^c}|\bP_A O_B(t))}_{\textstyle\Gamma_{il}}\,,
\end{align}
where $\bP_A$ is the projection superoperator \eqref{eq:proj_op_def}. The final expression is
\begin{subequations}\label{eq:altFormCinfDefMatrixElements}
\begin{align}
	\cinf_{A}(O_B) &= \sum_{jl}\Theta_{jl}\gamma_{jl},\\
    \gamma_{jl} &\equiv\sum_{i}\Gamma^{*}_{ij}\Gamma_{il},\label{eq:altFormCinfDefMatrixElementsgamma}\\
    \Gamma_{ij} &= (\sigma^i_A \sigma^j_{A^c}|\bP_A O_B(t)),\\
	\Theta_{jl} &= \frac{1}{d_A^2-1}\left[d_A\tr\big(\sigma^{j}_{A^c}\rho_{A^c}\sigma^l_{A^c}\rho_{A^c}\big) - \tr\big(\sigma^j_{A^c}\rho_{A^c}\big)\tr\big(\sigma^l_{A^c}\rho_{A^c}\big)\right],\label{eq:altMatElePsi}
\end{align}
\end{subequations}
or in operator form,
\begin{subequations}\label{eq:altFormCinfDef}
\begin{align}
    \cinf_{A}(O_B) &= \tr\left(\Theta\gamma\right),\\
    \gamma &= \tr_A\big[|\Gamma)(\Gamma|\big],\label{eq:altFormCinfDefgamma}\\
    |\Gamma) &= \bP_A|O_B(t))\,.
\end{align}
\end{subequations}
The two sets of equations \eqref{eq:altFormCinfDefMatrixElements} and \eqref{eq:altFormCinfDef} are equivalent.

To obtain an expression more suited for numerical purposes, we can explicitly carry out the Pauli sums in \eqref{eq:altFormCinfDefMatrixElementsgamma}, or equivalently, the partial trace in \eqref{eq:altFormCinfDefgamma}. This gives
\begin{align}
    \cinf_{A}(O_B) &= \frac{1}{d_A^2 - 1}\bigg\{\tr\Big(\big(O_B(t)\rho_{A^c}\big)^2\Big) - \frac{1}{d_A}\tr\Big(\tr_{A^c}\big(O_B(t)\rho_{A^c}\big)^2\Big) - \frac{1}{d_A}\tr\Big(\tr_{A}\big(O_B(t)\rho_{A^c}\big)^2\Big) + \frac{1}{d^2_A}\tr\big(O_B(t)\rho_{A^c}\big)^2\bigg\},
\end{align}
where $\rho_{A^c} \equiv \idm_A\otimes\tr_A(\rho)$.

\subsection{Properties and the spectrum of $\Theta$}

\subsubsection{$\Theta$ and $\gamma$ are positive semidefinite}\label{app:psiIsPSD}

\begin{lemma}
    The superoperator $\gamma$ is positive semidefinite.
\end{lemma}
\begin{proof}
    This is manifest from its operator form, $\gamma = \tr_A\big[|\Gamma)(\Gamma|\big]$.
\end{proof}

\begin{lemma}
    If $d_A \equiv \dim\cH_A \leq \dim\cH_{A^c}$ and $\rho\in\mathrm{End}(\cH)$ is a pure state, then the superoperator $\Theta$ is positive semidefinite.
\end{lemma}
\begin{proof}
Define $P_{A^c} \equiv \sum_\alpha v_\alpha \sigma_{A^c}^\alpha$. It suffices to show that 
\begin{align}
d_A\,\tr\!\left(P_{A^c}^\dag\rho_{A^c}P_{A^c}\rho_{A^c}\right) - \abs{\tr\left(P_{A^c}\rho_{A^c}\right)}^2 \geq 0
\end{align}
for all $P_{A^c}\in\mathrm{End}(\cH_{A^c})$. Let the spectral decomposition of $\rho_{A^c}$ be
\begin{align}
    \rho_{A^c} &= \sum_{k=1}^r p_k \ket{\phi_k}\bra{\phi_k},
\end{align}
where $r \leq \min(\dim\cH_A, \dim\cH_{A^c}) = d_A$. Define
\begin{align}
    Q &\equiv \sqrt{\rho_{A^c}} P_{A^c} \sqrt{\rho_{A^c}},\quad \Pi \equiv \sum_{k=1}^r \ket{\phi_k}\bra{\phi_k}.
\end{align}
By the Cauchy-Schwarz inequality, we have
\begin{align}
    d_A\tr\!\left(P_{A^c}^\dag\rho_{A^c}P_{A^c}\rho_{A^c}\right) - \abs{\tr\left(P_{A^c}\rho_{A^c}\right)}^2 &= d_A \tr(Q^\dag Q) - |\tr(\Pi Q)|^2 \geq d_A \tr(Q^\dag Q) - \tr(\Pi)\tr(Q^\dag Q) \geq 0.
\end{align}
\end{proof}

\subsubsection{Spectrum of $\Theta$ if $d_A = 2$}

In this section and the next, we compute the spectrum of $\Theta$, relating it to the pure state $\rho\in\mathrm{End}(\cH)$. It will be useful to view $\Theta$ not only as a matrix per its original definition in \eqref{eq:altFormCinfDefMatrixElements}, but as a superoperator that acts on the space $\mathrm{End}(\cH_{A^c})$. 

The object $\Theta_{ij}$ is manifestly a $d_{A^c}^2\times d_{A^c}^2$ matrix of real entries. As a superoperator, its matrix elements in the basis of Pauli strings are
\begin{align}\label{eq:psiIJ_spec}
	\Theta_{ij} &= (\sigma^i_{A^c}|\Theta|\sigma^j_{A^c}).
\end{align}
To compute the spectrum of $\Theta$, we solve the eigenvalue equation $\lambda v_i = \sum_j\Theta_{ij}v_j$. Using \eqref{eq:altMatElePsi} and \eqref{eq:psiIJ_spec}, the eigenvalue equation becomes
\begin{align}\label{eq:psiOpFormEigvalEq}
	\lambda v_{A^c} &= \Theta(v_{A^c}) \equiv \frac{d}{d_A^2-1}\Big(\rho_{A^c}v_{A^c}\rho_{A^c} - \frac{1}{d_A}\tr(\rho_{A^c}v_{A^c})\rho_{A^c}\Big)
\end{align}
where $v_{A^c} \equiv \sum_i v_i \sigma^i_{A^c}$. Now, suppose that $A$ consists of one qubit, so that $d_A = 2$. The spectral decomposition of $\rho_{A^c}$ is then
\begin{align}\label{eq:rhoAcDecomp}
	\rho_{A^c} &= p_1 \ket{\phi_1}\bra{\phi_1} + p_2 \ket{\phi_2}\bra{\phi_2}
\end{align}
where $\ket{\phi_1}, \ket{\phi_2} \in \cH_{A^c}$ are orthogonal. Then, we substitute \eqref{eq:rhoAcDecomp} and the ansatz $v_{A^c} = \alpha_1 \ket{\phi_1}\bra{\phi_1} + \alpha_2 \ket{\phi_1}\bra{\phi_2} + \alpha_3 \ket{\phi_2}\bra{\phi_1} + \alpha_4 \ket{\phi_2}\bra{\phi_2}$ into \eqref{eq:psiOpFormEigvalEq} and solve the resulting system. Ultimately, we find a particularly simple set of eigenvalues and eigenvectors, listed in \cref{table:eigvalsOfPsi}.

\begingroup
\setlength{\tabcolsep}{6pt} % Default value: 6pt
\renewcommand{\arraystretch}{1.25} % Default value: 1
\begin{table}
\centering
\begin{tabular}{|c|c|c|}
 \hline
 \multicolumn{3}{|c|}{$\Theta$} \\
 \hline
 Eigenvalue & Eigenvectors & Conditions \\
 \hline
  $\frac{d}{d_A^2-1}p_j p_k$ & $\Big\{\ket{\phi_j}\bra{\phi_k}, \ket{\phi_k}\bra{\phi_j}\Big\}$ & $ d_A \geq 2, j \neq k$ \\ [2.5ex] 
 $\frac{d}{3} \frac{\tr(\rho_A^2)}{2}$ & $ p_1\ket{\phi_1}\bra{\phi_1} - p_2\ket{\phi_2}\bra{\phi_2}$ & $ d_A = 2$ \\ [2.5ex]
  $0$ & $p_2\ket{\phi_1}\bra{\phi_1} + p_1\ket{\phi_2}\bra{\phi_2}$ & $d_A = 2$ \\ [2ex] 
 \hline
\end{tabular}
\begin{tabular}{|c|c|c|c|c|}
\hline
 \multicolumn{5}{|c|}{Properties (inside their respectively restricted domains)} \\
 \hline
 Operator & Dim & Rank & Nullity & Trace\\
 \hline
 $\Theta$ & $r^2$ & $r^2 - 1$ & $1$ & $\frac{d}{d_A^2-1}\left(1 - \|\rho_A\|^2_F\right)$ \\
 $\Theta_{\mathrm{od}}$ & $r(r-1)$ & $r(r - 1)$ & 0 & $\frac{d}{d_A^2-1}\left(1 - d_A\|\rho_A\|^2_F\right)$\\
 $\Theta'$ & $r$ & $r$ - 1 & $1$ & $\frac{d}{d_A-1}\|\rho_A\|^2_F$\\
 \hline
\end{tabular}
\caption[]{Left table: Eigenvalues and eigenvectors of $\Theta$, where $\rho_A = \tr_{A^c}(|\Psi\rangle\langle\Psi|)$ and $\rho_{A^c}$ has the decomposition \eqref{eq:rhoAcDecomp}. Right table: The (superoperator) rank and trace for the operators $\Theta$, $\Theta_{\mathrm{od}}$, and $\Theta'$, where $r \equiv \rank(\rho_A)$.}
\label{table:eigvalsOfPsi}
\end{table}
\endgroup

In operator form, $\Theta$ is given by
\begin{align}
    \Theta &= \frac{d^2}{6}\left(\frac{1}{2}|v^1_{A^c})(v^1_{A^c}| + p_1p_2\left(|v^2_{A^c})(v^2_{A^c}| + |v^3_{A^c})(v^3_{A^c}|\right)\right)
\end{align}
where
\begin{align}
    v^1_{A^c} \equiv p_1\ket{\phi_1}\bra{\phi_1} - p_2\ket{\phi_2}\bra{\phi_2}, \quad v^2_{A^c} \equiv \ket{\phi_1}\bra{\phi_2}, \quad v^3_{A^c} \equiv \ket{\phi_2}\bra{\phi_1},
\end{align}
are the eigenvectors of $\Theta$ given by \cref{table:eigvalsOfPsi}.
If, for example, the system is in the product state $\rho = \rho_A \otimes \rho_{A^c}$, then $\Theta$ has the matrix elements
\begin{align}\label{eq:rhoA-pure-psi}
    (\xi|\Theta|\xi) &= \frac{1}{3}|\tr(\rho_{A^c}\xi)|^2,\quad (d_A = 2, \rho = \rho_A \otimes \rho_{A^c})
\end{align}
where $\xi\in\mathrm{End}(\cH_{A^c})$.

\subsubsection{Spectrum of $\Theta$ if $d_A \geq 2$}

In general, for $d_A \geq 2$, the spectral decomposition of $\rho_{A^c}$ is
\begin{align}
    \rho_{A^c} = \sum_{k=1}^{\min(d_A, d_{A^c})} p_k \ket{\phi_k}\bra{\phi_k}
\end{align}
where again $\ket{\phi_k} \in \cH_{A^c}$ are orthogonal. The larger dimension complicates the eigenvectors and eigenvalues of $\Theta$ considerably, but parts of the spectrum remain straightforward. For example, $v_{A^c} = \ket{\phi_j}\bra{\phi_k}$ for $j \neq k$ is an eigenvector with eigenvalue $\frac{d}{d_A^2-1}p_j p_k$. 

In general, a nontrivial eigenvector of $\Theta$ can be expressed as the linear combination
\begin{align}
	v_{A^c}^{i} &= \sum_{j,k=1}^{\rank(\rho_A)}v^{i,jk} \ket{\phi_j}\bra{\phi_k}_{A^c}.
\end{align}
Note that it is sufficient to only include contributions from the eigenbasis of $\rho_{A^c}$, $\{\ket{\phi_j} : 1 \leq j \leq \rank(\rho_A)\}$, because the states orthogonal to the support of $\rho_{A^c}$ are canceled in \eqref{eq:psiOpFormEigvalEq}, giving an eigenvalue of $\lambda = 0$.

To continue, we define $\Theta = \Theta' +\Theta_{\mathrm{od}}$, where
\begin{align}
	\Theta_{\mathrm{od}}(u_{A^c}) \equiv \frac{d}{d_A^2-1}\sum_{j\neq k}p_j p_k\braket{\phi_j|u_{A^c}|\phi_k}\ket{\phi_j}\bra{\phi_k}_{A^c}
\end{align}
is the superoperator restricted to the ``off-diagonal'' part of $\Theta$ with eigenvalues given in the first row of \cref{table:eigvalsOfPsi}. The remaining eigenvectors, those of $\Theta'$, are spanned by the ``diagonal'' elements $\ket{\phi_i}\bra{\phi_i}$,
\begin{align}
	\Theta'(u_{A^c}) = \sum_{i=1}^{\rank(\rho_A)} \lambda_i \braket{\phi_i|u_{A^c}|\phi_i}\ket{\phi_i}\bra{\phi_i}_{A^c}.
\end{align}
For the sake of simplicity, we restrict $\Theta$ to $\mathrm{span}\{\ket{\phi_i}\bra{\phi_j} : 1 \leq i, j \leq \rank(\rho_A)\}$, $\Theta'$ to $\mathrm{span}\{\ket{\phi_i}\bra{\phi_i} : 1 \leq i \leq \rank(\rho_A)\}$, and $\Theta_{\mathrm{od}}$ to $\mathrm{span}\{\ket{\phi_i}\bra{\phi_j} : 1 \leq i \neq j \leq \rank(\rho_A)\}$. Their various properties are summarized in \cref{table:eigvalsOfPsi}. Finally, it is useful to note the matrix elements of $\Theta'$ in the $\ket{\phi_i}\bra{\phi_i}$ basis,
\begin{align}
	\Theta'_{ij} &\equiv \braket{\phi_i|\Theta\left(\ket{\phi_j}\bra{\phi_j}\right)|\phi_i} = \frac{d}{d_A^2 - 1}\left(\delta_{ij}p_j^2 - \frac{1}{d_A}p_ip_j\right).
\end{align}
In virtually all examples we examine, $d_A = 2$, in which case $\Theta'_{ij} = (d/3)(\delta_{ij} - 1/2)p_i p_j$ for $i, j\in \{1, 2\}$.

\subsection{Physical interpretation of the alternate form of the causal influence}\label{sec:physAltForm}

In \cref{app:altFormDerive} we introduced the $O_B$-dependent causal influence $\cinf_{A}(O_B)$, which captures how strongly perturbations on $A$ can influence the outcome of a measurement of $O_B$. It was shown to admit the compact `alternate form'
\begin{equation}
\cinf_{A}(O_B) = \underset{V_A}{\mathrm{Var}}\big[\tr\bigl(V_A^{\dagger}O_B(t)V_A\rho\bigr)\big]
=\tr(\Theta\gamma(t)),\label{eq:AltForm}
\end{equation}
where the two positive semidefinite operators,
\begin{subequations}
\begin{align}
\Theta &\equiv \frac{1}{d_A^2-1}\sum_{j,l}\left[d_A\tr\bigl(\sigma^{j}_{A^c}\rho_{A^c}\sigma^{l}_{A^c}\rho_{A^c}\bigr)
-\tr\bigl(\sigma^{j}_{A^c}\rho_{A^c}\bigr)\tr\bigl(\sigma^{l}_{A^c}\rho_{A^c}\bigr)\right]|\sigma^{j}_{A^c})(\sigma^{l}_{A^c}|,\label{eq:Psidef}\\
\gamma(t) &\equiv\tr_A\bigl[|\Gamma)(\Gamma|], \quad |\Gamma)=\bP_A|O_B(t)),
\end{align}
\end{subequations}
play complementary roles:
\begin{itemize}
\item \emph{Static factor $\Theta$} -- depends only on the state $\rho$ through its reduction $\rho_{A^c}$.
\item \emph{Dynamic factor $\gamma(t)$} -- depends only on the time-evolved operator $O_{B}(t)$ and therefore encodes the Hamiltonian and the evolution time $t$, but it is independent of $\rho$.
\end{itemize}
Either factor can suppress the causal influence, as the trace in~\eqref{eq:AltForm} vanishes whenever $\Theta$ and $\gamma(t)$ have orthogonal support. The rest of this subsection extracts the physical meaning of each factor, showing how their alignment—quantified by a spectral overlap—controls the flow of information from $A$ to $B$.

\subsubsection{State dependence: The operator $\Theta$ and entanglement}\label{subsec:PsiPhysics}

The operator $\Theta$ quantifies the innate sensitivity of the state to local perturbations on subsystem $A$. Equation~\eqref{eq:Psidef} shows that $\Theta$ is built solely from second-order moments of the reduced density matrix $\rho_{A^c}$. As proven in \cref{app:psiIsPSD}, $\Theta$ is positive semidefinite, which suggests an interpretation as a covariance matrix. Indeed, $\Theta_{jk}$ is the covariance of a specific set of random variables, describing the correlations between certain measurements that depend on the entanglement across the $A|A^c$ partition.

The connection to entanglement becomes explicit when we inspect the spectrum of $\Theta$. By diagonalizing this covariance matrix, we find its principal components: the directions corresponding to observables on $A^c$ that exhibit the largest variation when random unitary perturbations are applied to $A$. These eigenvectors $|v^k_{A^c})$ represent the system's most sensitive channels. Crucially, their corresponding eigenvalues $\lambda_k$, which quantify the sensitivity, are determined by the entanglement structure.

For instance, if subsystem $A$ is a qubit ($d_A=2$) and the global state is pure, the eigenvalues of $\rho_{A^c}$ are the Schmidt coefficients of the bipartite state. The most significant eigenvalues of $\Theta$ are directly proportional to the products of these Schmidt coefficients ($p_j p_k$ for $j \neq k$). If the state has no entanglement (i.e., it is a product state), then only one Schmidt coefficient is nonzero, and all of these key eigenvalues of $\Theta$ vanish. This provides a clear physical picture: \emph{entanglement opens sensitive channels through which information can flow, and $\Theta$ captures the structure of these channels}. Large eigenvalues $\lambda_k$ mark these highly sensitive directions, and their corresponding eigenvectors $|v^k_{A^c})$ define the operators on $A^c$ most susceptible to influence from perturbations on $A$.

\subsubsection{Dynamical dependence: The operator $\gamma$}\label{subsec:gammaPhysics}

Whereas $\Theta$ describes the state's static sensitivity, $\gamma(t)$ captures the dynamical spreading of the operator $O_B$ into the causal neighborhood of $A$. The calculation of $\gamma(t)$ begins with the Heisenberg-evolved operator $O_B(t)$. The unnormalized vector $|\Gamma)=\bP_A|O_B(t))$ projects $O_B(t)$ onto the subspace of operators that have nontrivial support on subsystem $A$. The subsequent partial trace over $A$ yields an operator that lives entirely on $A^c$.

Physically, $\gamma$ measures the extent to which the dynamics cause the operator $O_B$ to straddle the $A|A^c$ partition. For systems with local interactions, this takes time; thus, at early times or for weak couplings, the components of $O_B(t)$ with support on $A$ will be small, making $\gamma$ small. Conversely, chaotic dynamics can cause the operator to spread rapidly across the system, driving $\gamma$ towards its maximal value.

The spectrum of $\gamma$ reveals the ``directions'' of this operator spreading. If $O_B(t)$ evolves to have a large component proportional to an operator like $\sigma^i_A \otimes u^l_{A^c}$, then $|u^l_{A^c})$ will be an eigenvector of $\gamma$ with a large eigenvalue. In this way, \emph{$\gamma$ quantifies how strongly the dynamics couple operators on $A^c$ to the subsystem $A$ within the evolving operator $O_B(t)$}.

\vspace*{.2cm}
\emph{Example: a two-qubit interaction.}
To make this discussion of $\gamma(t)$ concrete, consider a simple system of two interacting qubits, labeled $A$ and $B$. Let subsystem $A$ be qubit $A$ and subsystem $A^c$ be qubit $B$. The dynamics are governed by the Ising Hamiltonian $H = J X_A Z_B$. We choose an initial observable $O_B = X_B$. The Heisenberg-evolved operator is then $O_B(t) = X_B\cos(2Jt) - X_A Y_B\sin(2Jt)$. It has two components: a part $\idm_A \otimes X_B\cos(2Jt)$ that acts trivially on $A$, and a part $-X_A \otimes Y_B\sin(2Jt)$ that acts on both $A$ and $B$. The projection operator $\bP_A$ isolates the part with nontrivial support on $A$: $\bP_A O_B(t) = -X_A \otimes Y_B\sin(2Jt)$.

We can now identify the operators $u^l_{A^c}(t)$ from the decomposition $\bP_A O_B(t) = \sum_{l=X,Y,Z} \sigma^l_A \otimes u^l_{A^c}(t)$. In this case, we have $u^x_{A^c}(t) = -Y_B\sin(2Jt)$, while $u^y_{A^c}(t)$ and $u^z_{A^c}(t)$ are zero. Using the definition of $\gamma(t)$, we find
\begin{align}
    \gamma(t) &= \sum_{l=X,Y,Z} |u^l_{A^c}(t))(u^l_{A^c}(t)| = |\!-Y_B\sin(2Jt))(-Y_B\sin(2Jt)| = \sin^2(2Jt) |Y_B)(Y_B|.
\end{align}
The final expression for $\gamma(t)$ is a rank-1 superoperator acting on $\mathrm{End}(\cH_B)$. Its time-dependent magnitude, $\sin^2(2Jt)$, quantifies the weight of the component of the evolved operator that has spread to act nontrivially on subsystem $A$. It starts at zero and grows as the interaction takes effect, illustrating how $\gamma(t)$ captures the flow of operator weight across the partition.

\subsubsection{Causal influence as spectral overlap}

The complete physical picture of the causal influence arises from combining the static, state-dependent sensitivity of $\Theta$ with the dynamic, operator-dependent spreading of $\gamma$. The expression $\cinf_{A}(O_B) = \tr(\Theta\gamma)$ is an inner product on the space of operators on $\mathcal{H}_{A^c}$, measuring the alignment between $\Theta$ and $\gamma(t)$. We can make this explicit by writing the trace in terms of the eigenvalues and eigenvectors of each operator. If $\Theta = \sum_k \lambda_k |v^k)(v^k|$ and $\gamma(t) = \sum_l \mu_l(t) |u^l(t))(u^l(t)|$, then:
\begin{align}
	\cinf_{A}(O_B) = \sum_{k,l} \lambda_k \mu_l(t) |(v^k|u^l(t))|^2\,.
\end{align}
This equation encapsulates the full story. For information to flow from $A$ to $B$, two conditions must be met simultaneously:
\begin{enumerate}
	\item The \emph{state} must be sensitive to perturbations along certain directions $|v^k)$ (i.e., have large eigenvalues $\lambda_k$).
	\item The \emph{dynamics} must spread the operator $O_B$ onto components $|u^l(t))$ that have significant weight (large $\mu_l(t)$) \emph{and} are aligned with the state's sensitive directions (large overlap $|(v^k|u^l(t))|^2$).
\end{enumerate}
Thus, the causal influence is strong only when the dynamical spreading of an operator occurs along the very channels that the entanglement structure of the state has made sensitive to perturbation.

\section{Proof of \cref{thm:acausal}}

In this section, we prove \cref{thm:acausal}, which provides the conditions for having exactly zero causal influence between two target regions. We rewrite the theorem here for convenience.
\begin{theorem*}
Let $A$ and $B$ be subsystems of one qubit each and let 
\begin{align}
    |\Psi\rangle\;=\; \sqrt{p_1} \ket{\psi_1}_A\ket{\phi_1}_{A^c} + \sqrt{p_2} \ket{\psi_2}_A\ket{\phi_2}_{A^c}
\end{align}
be the Schmidt form for the global state of the system, where $A^c$ is the complement of $A$. For each $\alpha\in\{X, Y, Z\}$, expand the Heisenberg-evolved Pauli at subsystem $B$ as
\begin{align}
\sigma_B^{\alpha}(t)
= \idm_A \otimes \nu^{\alpha0}_{A^c}(t) +\sum_{\beta=X, Y, Z}\sigma_A^{\beta} \otimes \nu^{\alpha\beta}_{A^c}(t).
\end{align}
Then the causal influence from $A$ to $B$ over a time $t$ vanishes, $\overline{\mathrm{CI}}_{AB}=0$, if and only if for every $\alpha,\beta\in\{X, Y, Z\}$
\begin{align}
p_1\langle\phi_1|\nu^{\alpha\beta}_{A^c}|\phi_1\rangle =p_2\langle\phi_2|\nu^{\alpha\beta}_{A^c}|\phi_2\rangle
\,\, \text{and}\,\,
\langle\phi_1|\nu^{\alpha\beta}_{A^c}|\phi_2\rangle&=0\,.
\end{align}
\end{theorem*}

\begin{proof}
First, recall the alternate form of the causal influence, $\cinf_A(O_B)$, defined by \eqref{eq:altFormCinfDef},
\begin{align}
    \cinf_{AB} = \ex_{O_B} \cinf_A(O_B) = \ex_{O_B} \tr(\Theta\gamma(t)),
\end{align}
where $\Theta$ and $\gamma(t)$ are positive semidefinite superoperators that act on $\mathrm{End}(\cH_{A^c})$. The expectation value is computed by sampling $O_B$ from the Hilbert-Schmidt measure. This measure has the property that $\ex_{O_B}[\tr(\Theta\gamma(t))] = 0$ if and only if $\tr(\Theta\gamma(t)) = 0$ for all $O_B$ that are positive semidefinite and have unit trace. The next step is to find the conditions equivalent to the latter statement.

Let $\Theta = \sum_k \lambda_k |v^k_{A^c})(v^k_{A^c}|$ and $\gamma(t) = \sum_l \mu_l(t) |u^l_{A^c}(t))(u^l_{A^c}(t)|$ be the spectral decompositions of $\Theta$ and $\gamma$. As a consequence of $\Theta$ and $\gamma$ having nonnegative eigenvalues,
\begin{align}
    \tr(\Theta\gamma(t)) = \sum_{k,l} \lambda_k \mu_l(t) |(v^k_{A^c}|u^l_{A^c}(t))|^2 = 0
\end{align}
if and only if $(v^k|u^l(t)) = 0$ for all $k, l$ with $\lambda_k, \mu_l(t) > 0$. Going forward, we will restrict this sum to only its nonzero terms.

One of our assumptions is that $A$ consists of a single qubit. Thus, \cref{table:eigvalsOfPsi} gives all the eigenoperators of the superoperator $\Theta$,
\begin{align}
    v^1_{A^c} = p_1\ket{\phi_1}\bra{\phi_1} - p_2\ket{\phi_2}\bra{\phi_2},\quad v^2_{A^c} = \ket{\phi_1}\bra{\phi_2}, \quad v^3_{A^c} = \ket{\phi_2}\bra{\phi_1},
\end{align}
where the $\ket{\phi_{1}}$ and $\ket{\phi_2}$ are the eigenstates of the reduced density matrix $\rho_{A^c}$. We obtain the following set of conditions by substituting each eigenoperator into $(v^k|u^l(t)) = 0$,
\begin{align}\label{eq:cinf_zero_conds_general}
     p_1\braket{\phi_1|u^l_{A^c}(t)|\phi_1} &= p_2\braket{\phi_2|u^l_{A^c}(t)|\phi_2}, \quad \braket{\phi_1|u^l_{A^c}(t)|\phi_2} = 0,\quad \text{for all $l$}.
\end{align}
By this point in the proof, we have shown that \eqref{eq:cinf_zero_conds_general} form a set of necessary and sufficient conditions for $\cinf_{AB} = 0$.

Now define the operators $\nu^{jl}_{A^c}(t)$ via the decomposition
\begin{align}\label{eq:thmProofStep}
    \bP_A \sigma_B^\alpha(t) = \sum_{l=X,Y,Z} \sigma^l_A \otimes \nu^{\alpha\beta}_{A^c}(t),
\end{align}
and write
\begin{align}
    O_B &= c_\idm\idm_B + c_X X_B + c_Y Y_B + c_Z Z_B.
\end{align}
By computing $\bP_A O_B(t)$ using these two equations, we obtain the eigenoperators of $\gamma$,
\begin{align}\label{eq:thmProofStep2}
    u^l_{A^c}(t) &= \sum_{j=X,Y,Z} c_j \nu^{jl}_{A^c}(t).
\end{align}
Clearly, the condition that \eqref{eq:cinf_zero_conds_general} hold for all $l$ and $O_B$ includes the choice where we set all except one of the $c_j$'s to zero in \eqref{eq:thmProofStep}. By substituting \eqref{eq:thmProofStep2} into \eqref{eq:cinf_zero_conds_general} we obtain
\begin{align}\label{eq:thmProofStep3}
    p_1\braket{\phi_1|\nu^{jl}_{A^c}(t)|\phi_1} &= p_2\braket{\phi_2|\nu^{jl}_{A^c}(t)|\phi_2},\quad\braket{\phi_1|\nu^{jl}_{A^c}(t)|\phi_2} = 0.
\end{align}
for all $l$ and that specific $j$. In the other direction, if \eqref{eq:thmProofStep3} holds for all $j, l$, then by linearity we know that \eqref{eq:cinf_zero_conds_general} holds for all $l$ and $O_B$.

The proof is mostly complete, but there is one subtlety remaining: technically, we have proven the theorem but with the $O_B$ restricted to be positive semidefinite and have unit trace. So the final part of our proof is to show that one can safely remove these restrictions. The key point is that we can choose an operator $O_B$ whose minimum eigenvalue is $-c$ (with $c>0$), and form the new operator
\begin{align}
    \tilde{O}_B = \frac{c\idm_B + O_B}{d_B c + \tr(O_B)}.
\end{align}
This one, unlike $O_B$, is postive semidefinite and has $\tr(\tilde{O}_B) = 1$. To complete the proof, one notes that $\cinf_A(\tilde{O}_B) = 0$ if and only if $\cinf_A(O_B) = 0$.
\end{proof}

\section{The error-corrected causal influence}

In this section, we define the \emph{error-corrected causal influence} (ECI), a quantity designed to diagnose the propagation of physical errors in a quantum error-correcting (QEC) code undergoing logical time evolution. Standard QEC focuses on correcting errors after they occur, but an analysis of the system's causal structure can reveal inherent protections where certain physical errors are incapable of affecting logical information. We aim to quantify how a local unitary perturbation $V_A$ on a physical qubit $A$ causally affects a logical observable $\overline{O}_{\overline{B}}$ in a logical region $\overline{B}$ after evolution under a logical Hamiltonian $\overline{H}$. 

The ECI, denoted $\cinf_{A\overline{B}}$, is a direct generalization of the standard causal influence. We define it as:
\begin{equation}
    \cinf_{A\overline{B}} \equiv \mathbb{E}_{\overline{O}_{\overline{B}}} \text{Var}_{V_A} \text{tr}(\overline{O}_{\overline{B}} \cR[\rho])
\end{equation}
where $\rho \equiv U(t)V_A \ket{\Psi}\bra{\Psi} V_A^\dag U(t)^\dag$. Here, $V_A$ is a random unitary operator drawn from the Haar measure on physical region $A$, $\overline{O}_{\overline{B}}$ is a random Hermitian operator on the logical qubits in region $\overline{B}$, $U(t) = \exp(-i\overline{H}t)$ is the time-evolution operator generated by the logical Hamiltonian, and $\cR$ is the QEC recovery map. 

A technical subtlety arises because the error-corrected state $\cR[\rho]$ is not guaranteed to be pure, a common assumption we make throughout the paper. We can resolve this by working in the Heisenberg picture, where the recovery map acts on the observable instead of the state. The expression inside the trace becomes $\text{tr}(\cR^\dag[\overline{O}_{\overline{B}}] \rho)$, which casts the causal influence as a response function to the effective operator $\cR^\dag[\overline{O}_{\overline{B}}]$. The adjoint of the recovery map for a stabilizer code is given by
\begin{align}
    \cR^{\dag}[O_{\overline B}] = \sum_{s}\Pi_s C_s^{\dag} O_{\overline B} C_s \Pi_s,
\end{align}
where $\Pi_s$ is the projector onto the syndrome-$s$ subspace and $C_s$ is the corresponding correction. 

This formulation ensures that our main results about the standard causal influence remain applicable. For example, the derivation of the alternate form of the causal influence in \cref{app:altFormDerive} depends on the Pauli expansion of the observable $O_B$. This derivation holds for the effective operator $\cR^\dag[\overline{O}_{\overline{B}}](t)$, as its logic is insensitive to the specific values of the Pauli expansion coefficients. It is true, however, that the time evolution of the effective operator is not norm-preserving (i.e., $\tr(\cR^\dag[\overline O_{\overline{B}}](t)^2) \neq \tr(\overline O_{\overline{B}}^2)$), and so the causal influence will be scaled correspondingly. Fortunately, this does not affect the conclusions of \cref{thm:acausal}, which is independent of the observable's normalization. Indeed, by making the replacement
\begin{align}
	\sigma^\alpha_x(\Delta t) &\mapsto \cR^\dag[\overline \sigma^\alpha_x](\Delta t) = \idm_q\otimes\nu^{\alpha0}_{q^c}(\Delta t) +\sum_{\beta=X,Y,Z} \sigma^\beta_q \otimes \nu^{\alpha\beta}_{q^c}(\Delta t)
\end{align}
in \eqref{eq:thm_state_decomp}, we recover \cref{thm:acausal} for the ECI. In this new version of the theorem, $x$ labels a \emph{logical qubit}, and $q$ a \emph{physical qubit}. Let us now proceed to apply \cref{thm:acausal} for the logical Hamiltonian evolution presented in the main text.

\subsection{The perfect $[[5,1,3]]$ code example}

To provide a concrete example of physical-to-logical causal protection, we apply \cref{thm:acausal} to a system of $k$ logical qubits. Each logical qubit is encoded using the perfect $[[5,1,3]]$ quantum error-correcting code. The goal is to identify initial states $|\Psi\rangle$ for which a single-qubit physical error on one logical block has no causal influence on another block after a time step $t$.

The five-qubit code is defined by the stabilizer generators
\begin{align*}
	X_{5j-4}Z_{5j-3}Z_{5j-2}X_{5j-1}\idm_{5j}, \quad \idm_{5j-4}X_{5j-3}Z_{5j-2}Z_{5j-1}X_{5j} \\
	X_{5j-4}\idm_{5j-3}X_{5j-2}Z_{5j-1}Z_{5j}, \quad Z_{5j-4}X_{5j-3}\idm_{5j-2}X_{5j-1}Z_{5j}
\end{align*}
for the $j$-th block of physical qubits, where $j=1,2, \dots, k$. The corresponding logical Pauli operators are products of single-qubit Paulis over three qubits of each block:
\begin{align}
	\overline{X}_j &= X_{5j-4} Y_{5j-2} Y_{5j-1}\\
	\overline{Z}_j &= Z_{5j-4}Y_{5j-3}Y_{5j-2}
\end{align}

We consider a nearest-neighbor logical Hamiltonian
\begin{align}
    \overline{H} = \sum_{j=1}^{k-1} \overline X_j \overline X_{j+1}
\end{align}
\cref{thm:acausal} requires us to evaluate matrix elements of the operators
\begin{align}
	\nu^{\alpha\beta}_{q^c}(t) &= \frac{1}{2}\,\tr_q\!\left(\sigma^\beta_q \cR^\dag[\overline \sigma^\alpha_x](t)\right),
\end{align}
where the adjoint of the recovery map $\cR^\dag$ for a stabilizer code acts on a logical operator $\overline{O}$ as
\begin{align}
    \cR^\dag[\overline{O}] = \Pi_{\text{code}} \overline{O} \Pi_{\text{code}} + \sum_{s \neq 0} \Pi_s C_s^\dag \overline{O} C_s \Pi_s,
\end{align}
where $\Pi_s$ projects onto the subspace with error syndrome $s$ and $\Pi_0 = \Pi_{\mathrm{code}}$ is the code space projector. 

Each syndrome $s$ in the recovery operation corresponds to a specific single-qubit error $C_s$. Suppose we take $\overline O$ to be a logical Pauli operator. Both being Pauli operators, $C_s$ and $\overline O$ can either commute or anticommute. We can use the completeness identity $\idm = \sum_s \Pi_s$ to eliminate all the commuting error terms from the sum,
\begin{align}
    \cR^\dag[\overline{O}] &= \Pi_{0} \overline{O} \Pi_{0} + \sum_{s\,:\, [C_s, \overline O] = 0} \Pi_s \overline O \Pi_s - \sum_{s\,:\, \{C_s, \overline O\} = 0} \Pi_s \overline{O} \Pi_s\\
    &= \overline{O}\left(\Pi_{0} + \sum_{s\,:\, [C_s, \overline O] = 0} \Pi_s - \sum_{s\,:\, \{C_s, \overline O\} = 0} \Pi_s \right)\\
    &= \overline{O}\left(\idm - 2\sum_{s\,:\, \{C_s, \overline O\} = 0} \Pi_s \right),\label{eq:recOpsQEC}
\end{align}
where in the second line we used the fact that any logical operator must commute with the stabilizers.

Let us take $q$ to be a physical site in a code block distinct from our ``target block'' $j$. Our goal is to identify states for which there is no causal influence from physical qubit $q$ on logical qubit $j$. To do this, we first perform a simplification of $\nu^{\alpha\beta}_{q^c}$ using properties of the code. This simplification starts by recalling that the conditions of \cref{thm:acausal} depend on matrix elements of the form $\braket{\phi_a | \nu_{q^c} | \phi_b}$, where $\ket{\phi_a}$ are defined in \cref{thm:acausal} by \eqref{eq:thm_schmidt}. With this in mind, let us break down the simplification into a few different steps:
\begin{enumerate}
	\item The $\ket{\phi_i}$ are states on $q^c$, so we may commute them with a partial trace over $q$. Applied to the relevant matrix elements, we have
	\begin{align}
		\braket{\phi_a | \nu^{\alpha\beta}_{q^c}(t) | \phi_b} &= \frac{1}{2}\,\tr_q\!\left(\sigma^\beta_q \braket{\phi_a|\cR^\dag[\overline \sigma^\alpha_j](t)|\phi_b}\right).
	\end{align}
	\item The evolution operator $U(t) = \exp(-i\overline{H}t)$ commutes with every stabilizer and therefore commutes with all syndrome projectors $\Pi_s$. This means that we can commute $U(t)$ through the summation in \eqref{eq:recOpsQEC}, obtaining $\cR^\dag[\overline \sigma^\alpha_j](t) = \overline\sigma^\alpha_j(t)\left(\idm - 2\sum_{s\,:\, \{C_s, \overline\sigma^\alpha_j\} = 0} \Pi_s \right)$.
	\item The matrix element $\braket{\phi_a | \nu^{\alpha\beta}_{q^c}(t) | \phi_b}$ now contains the terms $\Pi_s\ket{\phi_b}$ for syndrome $s$ with $\{C_s, \overline\sigma^\alpha_j\} = 0$. The anticommutation of $C_s$ with $\overline\sigma^\alpha_j$ implies that $C_s$ acts trivially outside of block $j$, and therefore so does $\Pi_s$. Thus, the $\Pi_s$ commute with the partial trace over $q$, which resides in a different block.
    
    Noting that $|\phi_b\rangle$ are eigenvectors of the reduced state $\rho_{q^c} = \tr_q(\ket{\Psi}\bra{\Psi})$, we have
	\begin{align}
		\Pi_s \rho_{q^c} = \tr_q(\Pi_s\ket{\Psi}\bra{\Psi}) = 0,
	\end{align}
    where the second equality uses the fact that $\ket{\Psi}$ is a logical state. This implies that the support of $\rho_{q^c}$ is orthogonal to all error subspaces of block $j$, and thus the Schmidt vectors are also annihilated: $\Pi_s \ket{\phi_b} = 0$. We conclude that the $\braket{\phi_a | \nu^{\alpha\beta}_{q^c}(t) | \phi_b}$ do not depend on the recovery process $\cR$.
\end{enumerate}
To summarize the simplification, we step-by-step justified the following string of equalities:
\begin{align}
	\braket{\phi_a | \nu^{\alpha\beta}_{q^c}(t) | \phi_b} &= \frac{1}{2}\,\tr_q\!\left(\sigma^\beta_q \braket{\phi_a|\cR^\dag[\overline \sigma^\alpha_j](t)|\phi_b}\right)\\
	&= \frac{1}{2}\,\tr_q\!\left(\sigma^\beta_q \braket{\phi_a|\overline\sigma^\alpha_j(t)\left(\idm - 2\textstyle\sum_{s\,:\, \{C_s, \overline\sigma^\alpha_j\} = 0} \Pi_s \right)|\phi_b}\right)\\
	&= \frac{1}{2}\,\tr_q\!\left(\sigma^\beta_q \braket{\phi_a|\overline \sigma^\alpha_j(t)|\phi_b}\right)
\end{align}
We can now pull out the $\braket{\phi_a|\cdot|\phi_b}$ and write
\begin{align}\label{eq:nuQEC}
	\nu^{\alpha\beta}_{q^c}(t) &= \frac{1}{2}\,\tr_q\!\left(\sigma^\beta_q\overline \sigma^\alpha_j(t)\right).
\end{align}
The next step is to compute the Heisenberg evolution of the logical Pauli operators $\overline \sigma^\alpha_j$.

The Heisenberg evolution of the logical Pauli operators under $\overline{H}$ mirrors that of the 1D transverse-field Ising model. We have
\begin{align}
\overline{Y}_j(t) &= \cos^2(2t)\overline{Y}_j - \sin^2(2t)\overline{X}_{j-1}\overline{Y}_j\overline{X}_{j+1} - \sin(2t)\cos(2t)(\overline{X}_{j-1}\overline{Z}_j + \overline{Z}_j\overline{X}_{j+1}) \\
\overline{Z}_j(t) &= \cos^2(2t)\overline{Z}_j - \sin^2(2t)\overline{X}_{j-1}\overline{Z}_j\overline{X}_{j+1} + \sin(2t)\cos(2t)(\overline{X}_{j-1}\overline{Y}_j + \overline{Y}_j\overline{X}_{j+1})
\end{align}
where we now ignore the initial operator choice of $\overline X_j$ because it is a symmetry of the Hamiltonian. WLOG we choose the logical qubit $j = 2$ and the physical qubit $q = 1$ located in the first code block.

Now we apply \eqref{eq:nuQEC} with $q = 1$, giving
\begin{align}
	\nu^{YX}_{1^c} &= -\frac{1}{2}\sin(4t) Y_3Y_4\left[\tan(2t)\overline{Y}_2\overline{X}_3 + \overline{Z}_2\right],\\
	\nu^{ZX}_{1^c} &= -\frac{1}{2}\sin(4t)Y_3Y_4\left[\tan(2t)\overline{Z}_2\overline{X}_3 - \overline{Y}_2\right].
\end{align}
Note that $Y_3 Y_4$ acts on the third and fourth \emph{physical} qubits of the first code block, and the barred operators act at the logical level.

If $t$ is an integer or half-integer multiple of $\pi$, then $\nu^{YX}_{1^c} = \nu^{ZX}_{1^c} = 0$ and the Theorem conditions are trivially satisfied by any $\ket{\Psi}$. Otherwise, we drop the shared prefactor and write
\begin{align}\label{eq:qecEigops}
\begin{split}
    \nu^{YX}_{1^c} &\propto Y_3Y_4\left[\tan(2t)\overline{Y}_2\overline{X}_3 + \overline{Z}_2\right],\\
	\nu^{ZX}_{1^c} &\propto Y_3Y_4\left[\tan(2t)\overline{Z}_2\overline{X}_3 - \overline{Y}_2\right].
\end{split}
\end{align}
Our final step is to use an ansatz for the initial state $\ket{\Psi}$ and find a family of solutions that satisfy all of
\begin{align}
p_1\langle\phi_1|\nu^{YX}_{1^c}|\phi_1\rangle =p_2\langle\phi_2|\nu^{YX}_{1^c}|\phi_2\rangle
\,\, \text{and}\,\,
\langle\phi_1|\nu^{YX}_{1^c}|\phi_2\rangle&=0,\\
p_1\langle\phi_1|\nu^{ZX}_{1^c}|\phi_1\rangle =p_2\langle\phi_2|\nu^{ZX}_{1^c}|\phi_2\rangle
\,\, \text{and}\,\,
\langle\phi_1|\nu^{ZX}_{1^c}|\phi_2\rangle&=0,
\end{align}
where $p_1, p_2$ and $\ket{\phi_1}$, $\ket{\phi_2}$ are defined in the Theorem statement.

We consider an initial state with no logical entanglement between the first block and the others:
$$|\Psi\rangle = (a|\overline{0}\rangle_1 + b|\overline{1}\rangle_1) \otimes |\Psi_{23}\rangle,$$
where $a, b \in \mathbb{C}$ and $|a|^2+|b|^2=1$. The Schmidt decomposition with respect to qubit $q=1$ yields two equally weighted Schmidt vectors $|\phi_1\rangle$ and $|\phi_2\rangle$. The diagonal conditions of \cref{thm:acausal} are automatically satisfied by the structure of the operators,
\begin{equation}
	\langle\phi_1|\nu^{\alpha\beta}_{1^c}|\phi_1\rangle = \langle\phi_2|\nu^{\alpha\beta}_{1^c}|\phi_2\rangle = 0
\end{equation}
The off-diagonal constraints are nontrivial, giving
\begin{align}
ab\,\langle\Psi_{23}|\left[\tan(2t)\,\overline{Y}_2\overline{X}_3 + \overline{Z}_2\right]|\Psi_{23}\rangle = 0,\quad\text{and}\quad
ab\,\langle\Psi_{23}|\left[\tan(2t)\,\overline{Z}_2\overline{X}_3 - \overline{Y}_2\right]|\Psi_{23}\rangle = 0
\end{align}
As discussed in the main manuscript, the solutions fall into three families:
\begin{enumerate}
	\item If $ab=0$, the first logical qubit is in a $\overline{Z}_1$ eigenstate. It has no quantum coherence to be disturbed, so no information can propagate.
	\item If $|\Psi_{23}\rangle$ is a product state where the second logical qubit is an eigenstate of $\overline{X}_2$ (e.g., $|\overline{\pm}\rangle_2|\Psi\rangle_3$), the expectation values vanish.
	\item Genuinely entangled states of blocks 2 and 3 can satisfy the conditions. Notable examples include the Bell-like state $\frac{1}{\sqrt{2}}(\ket{\overline{00}}+\ket{\overline{11}})$ and the one-parameter family $\frac{1}{\sqrt{2}}\bigl(\cos(t+\tfrac\pi4)\ket{\overline{00}}+i\cos(t-\tfrac\pi4)\ket{\overline{11}}\bigr)$.
\end{enumerate}
This example demonstrates that even when logical information is being dynamically scrambled, the structure of logical entanglement can create inherent protections against the propagation of physical errors, a feature that can be identified by the framework of \cref{thm:acausal}. In the next section, we will examine a symmetry in the solutions of this example that identify the logical entanglement relevant to acausal features.

\subsection{Symmetries in acausal logical states}

Having explored a concrete example of applying the theorem to compute vanishing ECI, we now examine a structural feature of the theorem: a symmetry arising from the algebraic properties of the $\nu^{\alpha\beta}_{A^c}$ operators. This symmetry serves as a valuable tool for simplifying the analysis of the solution states $\ket{\phi_1}$ and $\ket{\phi_2}$. In particular, it enables us to reduce any long-range entangled logical state solution in the example to a pure logical state confined to the code blocks $j-1, j, j+1$. It is the logical entanglement among these three blocks that determines the acausal features of the previous example, rather than long-range correlations with other parts of the system.

Consider a unitary operator $W\in\mathrm{End}(\cH_{A^c})$ that commutes with the operators of \cref{thm:acausal},
\begin{align}\label{eq:uniBlock}
	[W, \nu^{\alpha\beta}_{A^c}] &= 0,\quad\text{for all $\alpha,\beta$}.
\end{align}
We have 
\begin{align}
	p_1\langle\phi_1|\nu^{\alpha\beta}_{A^c}|\phi_1\rangle = p_2\langle\phi_2|\nu^{\alpha\beta}_{A^c}|\phi_2\rangle,\,\,\text{and}\,\,\braket{\phi_1|\nu^{\alpha\beta}_{A^c}|\phi_2} = 0
\end{align}
if and only if
\begin{align}\label{eq:solutsSymm}
	p_1\langle\phi_1| W^\dag\nu^{\alpha\beta}_{A^c}W|\phi_1\rangle =p_2\langle\phi_2|W^\dag\nu^{\alpha\beta}_{A^c}W|\phi_2\rangle,\,\,\text{and}\,\,\braket{\phi_1|W^\dag\nu^{\alpha\beta}_{A^c}W|\phi_2} = 0.
\end{align}

In our previous example, the operators under consideration were $\nu^{YX}_{A^c}$ and $\nu^{ZX}_{A^c}$ as given by \eqref{eq:qecEigops}. For convenience we rename them to $O$ and $Q$, respectively, and renormalize them such that each squares to the identity:
\begin{align}
\begin{split}\label{eq:thmSymOps}
	O &\equiv \frac{1}{r}Y_{5j-7}Y_{5j-6}\left[a\overline{Y}_j\overline{X}_{j+1} + \overline{Z}_j\right],\\
	Q &\equiv \frac{1}{r}Y_{5j-7}Y_{5j-6}\left[a\overline{Z}_j\overline{X}_{j+1} - \overline{Y}_j\right],
\end{split}
\end{align}
where $a \equiv \tan(2t)$ and $r \equiv \sqrt{1 + a^2}$.

In the following lemma, we prove that there is always some $W$ that allows us to disentangle a logical state solution of \cref{thm:acausal} into a logical product state across the support of $O$ and $Q$, and the rest of the system. Practically, the lemma provides formal justification for using solution ansatz restricted to the support of $O$ and $Q$.

\begin{lemma}\label{thm:logicalDisent1}
	For any logical state $\ket{\Psi}$ on the full system of three consecutive code blocks plus any additional qubits, there exists a unitary operator $\idm_{A}\otimes W$ (acting as identity on the first physical qubit $A$ of the first block and commuting with $O$ and $Q$) and a product logical state $\ket{\Phi}\otimes\ket{\chi}$ (with $\ket{\Phi}$ on the three consecutive code blocks and $\ket{\chi}$ on the additional qubits), such that
	\begin{align*}
	\left(\idm_{A}\otimes W\right)\ket{\Psi} &= \ket{\Phi}\otimes\ket{\chi}.
	\end{align*}
\end{lemma}

\begin{proof}
Let \(\ket{\Psi}\) be an arbitrary pure logical state in the total logical codespace \(L_\text{total} \equiv L_{j-1} \otimes L_j \otimes L_{j+1} \otimes L_\text{extra}\), where each \(L_k\) (for \(k = j-1, j, j+1\)) is the 2-dimensional logical subspace of the \([[5,1,3]]\) code block \(k\), and \(L_\text{extra}\) is the logical subspace of any additional code blocks. The goal is to prove the existence of a unitary \(\idm_{A} \otimes W\) on the underlying physical Hilbert space (spanned by all physical qubits across the blocks and extras), with $A$ the first physical qubit of block $j-1$, that commutes with \(O\) and \(Q\) and disentangles \(\ket{\Psi}\) into \(\ket{\Phi} \otimes \ket{\chi}\) for some pure \(\ket{\Phi} \in L_{j-1} \otimes L_j \otimes L_{j+1}\) and \(\ket{\chi} \in L_\text{extra}\).

The operators \(O\) and \(Q\) are defined on the subspace \(\mathcal{H}_\text{in} \equiv \bigotimes_{i=2}^5 \mathcal{H}_{q_i} \otimes L_j \otimes L_{j+1} \otimes L_\text{extra}\), where \(\bigotimes_{i=2}^5 \mathcal{H}_{q_i}\) is the physical space of qubits 2 through 5 in block \(j-1\). On this domain, \(O\) and \(Q\) act as the physical operator \(Y_{5j-7} Y_{5j-6}\) on the specified qubits in block \(j-1\), as logical operators \(\overline{O}_{j,j+1} \equiv (a \overline{Y}_j \overline{X}_{j+1} + \overline{Z}_j)/r\) and \(\overline{Q}_{j,j+1} \equiv (a \overline{Z}_j \overline{X}_{j+1} - \overline{Y}_j)/r\) on \(L_j \otimes L_{j+1}\), and as the identity on \(L_\text{extra}\). They extend to \(H_\text{total}\) by acting as the identity on \(\mathcal{H}_{A}\).

To proceed, note that \(O\), \(Q\), and \(\overline{X}_j\) form a Pauli-like algebra on \(\mathcal{H}_\text{in}\): $O^2=Q^2=\overline{X}_j^2=\idm$, with anticommutators $\{O, Q\}=\{Q, \overline{X}_j\}=\{\overline{X}_j, O\}=0$ and commutation relations $O Q=i \overline{X}_j, Q \overline{X}_j=i O, \overline{X}_j O=i Q$. This implies $\cH_{\mathrm{in}}$ decomposes as $\cH_{\mathrm{in}} \cong \cH_{\mathrm{eff}} \otimes \cK$ for a 2D effective qubit $\cH_{\mathrm{eff}}$ and auxiliary $\mathcal{K}$, where $O \equiv X_{\mathrm{eff}}\otimes\idm_\cK$, $Q \equiv Y_{\mathrm{eff}}\otimes\idm_\cK$, and $\overline X_j \equiv Z_{\mathrm{eff}}\otimes\idm_\cK$. Any unitary $W$ commuting with $O$ and $Q$ acts as $W = \idm_{\mathrm{eff}}\otimes W_\cK$. 

Thus, we must find $W_\mathcal{K}$, $\ket{\Phi}$, and $\ket{\chi}$ such that $(\idm_{A} \otimes \idm_\text{eff} \otimes W_\mathcal{K})\ket{\Psi} = \ket{\Phi} \otimes \ket{\chi}$. This is equivalent to finding $\ket{\Phi}\otimes\ket{\chi}$ such that $\rho_{A,\mathrm{eff}} = \tr_\cK(\ket{\Phi}\bra{\Phi}\otimes\ket{\chi}\bra{\chi})$, where $\rho_{A,\mathrm{eff}} \equiv \tr_\cK(\ket{\Psi}\bra{\Psi})$. In the next step, we will prove the existence of such a state $\ket{\Phi}$ by showing $\braket{\Psi|\sigma^\alpha_A\otimes\sigma^\beta_{\mathrm{eff}}\otimes\idm_\cK|\Psi} = \braket{\Phi|\sigma^\alpha_A\otimes\sigma^\beta_{\mathrm{eff}}\otimes\idm_\cK|\Phi}$ for all $\alpha, \beta\in\{X, Y, Z\}$.

Using that $\ket{\Psi}$ is a logical state, one can see (with a direct calculation) that there are exactly three nonzero expected values in $\cH_A\otimes\cH_\mathrm{eff}$; namely, $\braket{X_{A} X_\text{eff}}$, $\braket{X_{A} Y_\text{eff}}$, and $\braket{Z_\text{eff}}$. This is also true for the $\ket{\Phi}$ we eventually construct, because it too is logical. Hence we only need $\ket{\Phi}$ to agree with $\ket{\Psi}$ on those three operators---all other expectation values of the form $\braket{\sigma^\alpha_A\otimes\sigma^\beta_{\mathrm{eff}}\otimes\idm_\cK}$ are automatically zero on both states.

We now perform a trick using the fact that $X_A = X_{5j-9}$ together with $Y_{5j-7}Y_{5j-6}$ form the logical operator $\overline X_{j-1}$. More precisely, we write $\braket{X_{A} X_\text{eff}} = \braket{\overline{X}_{j-1} \overline{O}_{j,j+1}}$, $\braket{X_{A} Y_\text{eff}} = \braket{\overline{X}_{j-1} \overline{Q}_{j,j+1}}$, and $\braket{Z_\text{eff}} = \braket{\overline{X}_j}$, where the expectation value is taken with respect to any logical state. The operators $\{\idm, \overline X_{j-1} \overline O_{j,j+1}, \overline X_{j-1} \overline Q_{j,j+1}, \overline X_j\}$ generate a Pauli algebra on $\cH_A\otimes\cH_\mathrm{eff}$, yielding $\cH_A\otimes\cH_\mathrm{eff} \cong L_\mathrm{eff}\otimes \cK'$ where $L_\mathrm{eff}, \cK' \subset L_{j-1}\otimes L_j \otimes L_{j+1}$ are the \emph{logical} effective qubit and auxiliary subspaces, respectively.

We can apply our previous argument to this new effective qubit algebra: the expectation values from the previous paragraph match on both states if and only if $\tr_{\cK'}(\rho_{A,\mathrm{eff}}) = \tr_{\cK'}(\ket{\Phi} \bra{\Phi})$. Given that $\tr_{\cK'}(\rho_{A,\mathrm{eff}})$ is a $2\times2$ density matrix and $\dim \mathcal{K}' = 2$, the mixed state $\tr_{\cK'}(\rho_{A,\mathrm{eff}})$ admits a purification $\ket{\Phi} \in L_\text{eff} \otimes \mathcal{K}'$. Therefore, the state $\ket{\Phi}$ matches $\ket{\Psi}$ on the nonzero expectation values (i.e., $\braket{X_{A} X_\text{eff}}$, $\braket{X_{A} Y_\text{eff}}$, $\braket{X_{A} Z_\text{eff}}$) and vanishes on the others (i.e., all other $\braket{\sigma^\alpha_A\otimes\sigma^\beta_{\mathrm{eff}}\otimes\idm_\cK}$). We are free to choose any $\ket{\chi}\in L_\mathrm{extra}$, and then we finally obtain $\rho_{A,\mathrm{eff}} = \tr_\cK(\ket{\Phi}\bra{\Phi}\otimes\ket{\chi}\bra{\chi})$.
\end{proof}

\section{Causal influences of syndrome extraction and correction}

In this section, we consider general $[[n, k]]$ stabilizer codes that correct single-qubit errors. In such codes, we refer to the $n$ qubits encoding logical information as \emph{physical qubits}. Sometimes we need to bring in \emph{ancilla qubits} to facilitate quantum error correction processes. We are interested in computing causal influences exerted by logical or physical qubits on logical or ancilla qubits. For example, the physical-to-ancilla causal influence ($\cinf_{\phys,\anc}$) measures how much a unitary ``kick'' on a single physical qubit affects the state of the ancilla qubits. Physically, it quantifies the flow of information from a potential error location to the error-detection system.

The dynamics of this problem are taken to be that of a single cycle of syndrome measurement and error correction. This is implemented via the formal process of bringing in ancilla qubits, carefully coupling them to the physical qubits, and measuring the ancilla qubits. Based on the measurement result, we apply a correction operator to the physical qubits. Formally, we define two quantum channels:
\begin{align}
	\cR'[\rho] = \sum_{s,t} C_s \Pi_s \rho \Pi_t C_t^\dag \otimes \ket{s}\bra{t}_{\mathrm{anc}}\,,\quad \cR[\rho] &= \tr_{\mathrm{anc}} \cR'[\rho]\,,
\end{align}
where $\Pi_s$ projects onto the syndrome subspace and $C_s$ is the unitary correction operator applied given syndrome $s$.

The channel $\cR$ is the standard recovery channel for a stabilizer code. The channel $\cR'$ is a dilated form of the recovery channel, before the ancilla qubits are measured. Whether or not we measure the ancillas does not alter $\cR[\rho]$. Note that in this formal analysis, the implementation details of $\cR'$ are abstracted out; for a concrete example with an explicit quantum circuit, please refer to \cref{fig:qec_cycle}. The process of determining $C_s$ from syndrome $s$ is known as ``syndrome decoding.'' Such details are not relevant in the formal analysis here, but it is possible that the decoding algorithm becomes relevant in less-than-ideal scenarios, such as imperfect decoding. We leave such analysis to future work.

As a final note, we define
\begin{align}
    \cR'_{\mgate}[\rho] = \sum_{s} C_s \Pi_s \rho \Pi_s C_s^\dag \otimes \ket{s}\bra{s}_{\mathrm{anc}}
\end{align}
as the version of $\cR'$ where the ancillas are measured in the $\{\ket{s}\}$ basis.

\subsection{Stabilizer code: formal results}

We are interested in the following response functions:
\begin{align}
\begin{split}
    M_{LL}(V_L, O_L) &\equiv \tr\left(O_L \cR(V_L\rho V_L^\dag)\right)\\
    M_{L,\anc}(V_L, O_\anc) &\equiv \tr\left(O_\anc \cR'(V_L\rho V_L^\dag)\right) \\
    M_{\phys,L}(V_j, O_L) &\equiv \tr\left(O_L \cR(V_j\rho V_j^\dag)\right)\\
    M_{\phys,\anc}(V_j, O_\anc) &\equiv \tr\left(O_\anc \cR'(V_j\rho V_j^\dag)\right)
\end{split}
\end{align}
where $V_j$ is a physical unitary on some physical qubit $j$. The subscripts denote logical-logical, logical-ancilla, physical-logical, and physical-ancilla, respectively.

\subsubsection{Logical-to-Ancilla Causal Influence}
A perturbation $V_L$ applied within the logical codespace has no effect on the error syndrome. The response function is independent of $V_L$, leading to a vanishing causal influence:
\begin{equation}
    \cinf_{L,\anc} = 0 \quad\text{(pre- and post-measurement)}.
\end{equation}
Physically, this means the error correction process is properly designed to be insensitive to logical operations, preventing any information about logical gates from leaking to the syndrome measurement ancillas.

\subsubsection{Physical-to-Logical Causal Influence}
For a code that corrects all single-qubit errors, any such error is perfectly reversed by the recovery channel $\cR$. A kick $V_j$ on a physical qubit is therefore fully corrected, and no information about it propagates to the logical qubits. The causal influence is consequently zero:
\begin{equation}
    \cinf_{\phys,L} = 0 \quad\text{(pre- and post-measurement)}.
\end{equation}
This confirms that information about a correctable physical error does not propagate to the logical subspace.

\subsubsection{Logical-to-Logical Causal Influence}
The recovery channel acts as the identity on states already within the codespace, so $\cR(V_L\rho V_L^\dag) = V_L\rho V_L^\dag$. The resulting causal influence is nonzero and depends only on the dimension $D=2^k$ of the codespace:
\begin{align}
    \boxed{\cinf_{LL} = \frac{D - 1}{D^2(D^2+1)}\quad\text{(pre- and post-measurement)}}
\end{align}

\subsubsection{Physical-to-Ancilla Causal Influence}
This quantity measures how information about a physical error on a qubit propagates to the ancillas. We have
\begin{align}
\begin{split}
    M_{\phys,\anc}(V_j, O_\anc) &= \sum_{s, t} \tr\!\left(C_s \Pi_s V_j \rho V_j^\dag \Pi_t C_t^\dag \right) \braket{t|O_\anc|s}\,.
\end{split}
\end{align}

The first moment Haar average $\ex_{V_j} V_j \rho V_j^\dag$ is equivalent to a completely depolarizing channel acting on qubit $j$, which is a correctable error. Squaring and averaging over $O_\anc$ then gives
\begin{align}\label{eq:m_sq_outer}
\begin{split}
    \ex_{O_\anc} \left[\ex_{V_j} \left[M_{\phys,\anc}(V_j, O_\anc)\right]\right]^2 &= \frac{1}{16}\ex_{O_\anc}\left[\sum_{\alpha,\beta=0}^3 \braket{s(\sigma^\alpha_j)\otimes s(\sigma^\beta_j)|O_\anc^{\otimes2}|s(\sigma^\alpha_j)\otimes s(\sigma^\beta_j)}\right]\\
    &= a(D_\anc) + \frac{b(D_\anc)}{4}
\end{split}
\end{align}
where $a(D) \equiv (D^2+1)^{-1}$, $b(D) \equiv  a(D)D^{-1}$. This is valid pre- and post-syndrome measurement.

\paragraph*{Pre-measurement squared response function.}
The other part that goes into the causal influence is the squared response function:
\begin{align}\label{eq:checkpoint1}
\begin{split}
    \ex_{O_\anc} \left[M_{\phys,\anc}(V_j, O_\anc)^2\right] &= \sum_{s, t, p, q} \tr\!\left(\cR_{st}(V_j\rho V_j^\dag)\right)\tr\!\left(\cR_{pq}(V_j\rho V_j^\dag) \right)I_{stpq}
\end{split}
\end{align}
where $\cR_{st}(X) \equiv C_s \Pi_s X \Pi_t C_t^\dag$ and $I_{stpq} = a(D_\anc)\delta_{ts}\delta_{qp} + b(D_\anc)\delta_{tp}\delta_{qs}$. Evaluating this sum gives
\begin{align}\label{eq:m_sq_step1}
\begin{split}
    \ex_{O_\anc} \left[M_{\phys,\anc}^2\right] &= a(D_\anc)\left(\sum_{s} \tr\!\left(\cR_{ss}(V_j\rho V_j^\dag) \right)\right)^2 + b(D_\anc)\sum_{s, t} \left|\tr\!\left(\cR_{st}(V_j\rho V_j^\dag)\right)\right|^2\\
    &= a(D_\anc) + b(D_\anc)\sum_{s, t} \left|\tr\!\left(\cR_{st}(V_j\rho V_j^\dag)\right)\right|^2\\
\end{split}
\end{align}
where we used $\cR = \sum_s \cR_{ss}$ and that $\cR$ preserves the trace. The next step is to use the Knill-Laflamme conditions, which say $\Pi_0 E_s^\dag V_j \Pi_0 = \alpha_s(V_j)\Pi_0$ for some $\alpha_s(V_j) \in \mathbb{C}$. Here we defined $E_s \equiv C_s^\dag$ as the error corresponding to the correction $C_s$. Substituting this into the summand of \eqref{eq:m_sq_step1} gives
\begin{align}
\begin{split}
    \cR_{st}(V_j\rho V_j^\dag) &= \Pi_0 E_s^\dag V_j \Pi_0 \rho \Pi_0 V_j^\dag E_t \Pi_0 = \alpha_s(V_j)\alpha_t(V_j)^{*}\rho
\end{split}
\end{align}
and thus
\begin{align}
    \sum_{s, t} \left|\tr(\cR_{st}(V_j\rho V_j^\dag))\right|^2 &= \sum_{s, t} |\alpha_s(V_j)|^2|\alpha_t(V_j)|^2 = \left(\sum_s |\alpha_s(V_j)|^2\right)^2.
\end{align}
This final sum is evaluated like so:
\begin{align}
\begin{split}
    \sum_s |\alpha_s(V_j)|^2\rho &= \sum_s \left(\Pi_0 E_s^\dag V_j \Pi_0\right)^\dag \left(\Pi_0 E_s^\dag V_j \Pi_0\right)\rho\\
    &= \Pi_0 V_j^\dag \left(\sum_s E_s \Pi_0 E_s^\dag\right) V_j \Pi_0\rho\\
    &= \Pi_0 V_j^\dag V_j \Pi_0\rho\\
    &= \rho
\end{split}
\end{align}
where we used $\Pi_s = E_s\Pi_0 E_s^\dag$ and $\sum_s \Pi_s = \idm$. This proves $\sum_s |\alpha_s(V_j)|^2 = 1$. Hence, \eqref{eq:m_sq_step1} is independent of $V_j$ and we obtain
\begin{align}
    \ex_{V_j}\ex_{O_\anc} \left[M_{\phys,\anc}^2\right] &= a(D_\anc) + b(D_\anc).
\end{align}
Combining this result with \eqref{eq:m_sq_outer} we obtain
\begin{align}
    \boxed{\cinf_{\phys,\anc} = \frac{3}{4}\frac{1}{D_\anc(D_\anc^2 + 1)}\quad\text{(pre-measurement)}}
\end{align}
where $D_\anc$ is the dimension of the ancillary subsystem.

\paragraph*{Post-measurement squared response function.}
For the post-measurement case, the channel is modified such that only diagonal terms $s=t$ contribute. Starting from \cref{eq:checkpoint1}, we have (post-measurement),
\begin{align}
\begin{split}
    \ex_{O_\anc}\left[M^{\mgate}_{\phys,\anc}(V_j, O_\anc)^2\right] &= \sum_{s,t} \tr\left(\cR_{ss}(V_j\rho V_j^\dag)\right)\tr\left(\cR_{tt}(V_j\rho V_j^\dag) \right)I_{sstt}\\
    &= a(D_\anc)\left(\sum_{s} \tr\left(\cR_{ss}(V_j\rho V_j^\dag)\right)\right)^2 + b(D_\anc) \sum_{s} \tr\left(\cR_{ss}(V_j\rho V_j^\dag)\right)^2\\
    &= a(D_\anc) + b(D_\anc) \sum_{s} \tr\left(C_s \Pi_s V_j\rho V_j^\dag \Pi_s C_s^\dag)\right)^2\\
    &= a(D_\anc) + b(D_\anc) \sum_{s} \tr(\Pi_s V_j\rho V_j^\dag)^2
\end{split}
\end{align}
where we used $I_{sstt} = a(D_\anc) + b(D_\anc)\delta_{st}$ in the second equality. In the sum over $s$ only four syndromes can occur, one for each possible Pauli error on qubit $j$,
\begin{equation}
    \sum_{s} \tr(\Pi_s V_j\rho V_j^\dag)^2 = \sum_{\mu = 0,X,Y,Z}\tr(\Pi_{s(j,\mu)} V_j \rho V_j^\dag)^2
\end{equation}
Using the orthonormal operator basis $P_\mu=\sigma^\mu/\sqrt{2}$ for a single qubit,  expand $V_j = \sum_\mu v^\mu P_\mu$ and write $\tr(\Pi_{s(j,\mu)} V_j \rho V_j^\dag) = |v^\mu|^2/2$. Substituting this relation into the above sum gives
\begin{equation}
    \sum_{s} \tr(\Pi_s V_j\rho V_j^\dag)^2 = \sum_\mu \left(\frac{|v^\mu|}{\sqrt{2}}\right)^4.
\end{equation}
We also have
\begin{equation}
    \frac{1}{2}\sum_\mu |v^\mu|^2 = 1.
\end{equation}
We now associate each $v^\mu/\sqrt{2}$ with the component of a unit quaternion. The Haar average over $V_j$, then, is equivalent to computing a uniform average over the unit quaternions. Extracting the fourth-order moments from that distribution gives
\begin{align}
    \ex_{V_j} \sum_{s} \tr(\Pi_s V_j\rho V_j^\dag)^2 = 4 \cdot \frac{1}{8} = \frac{1}{2}
\end{align}
and so
\begin{align}
    \ex_{V_j} \ex_{O_\anc} \left[M^{\mgate}_{\phys,\anc}(V_j, O_\anc)^2\right] = a(D_\anc) + \frac{b(D_\anc)}{2}
\end{align}
Combining this result with \eqref{eq:m_sq_outer} we obtain
\begin{equation}
    \boxed{\cinf^{\mgate}_{\phys,\anc} = \frac{1}{4}\frac{1}{D_\anc(D_\anc^2 + 1)}\quad\text{(post-measurement)}}
\end{equation}
The causal influence is smaller after measurement, a result that has a clear physical interpretation. The pre-measurement influence $\cinf_{\phys,\anc}$ quantifies the full quantum information channel between the physical and ancilla qubits. This influence is sensitive to the quantum coherences between different syndrome states (i.e., the off-diagonal terms $\ket{s}\bra{t}$). The act of measurement destroys these coherences, collapsing the superposition into a single classical outcome. By discarding this quantum information, the post-measurement ancillas are less sensitive to the kick, and the measured influence $\cinf^{\mgate}_{\phys,\anc}$ is correspondingly weaker.

\subsection{Example: Repetition code}

The repetition code provides an instructive contrast because it can only correct one type of error (e.g., phase-flips) but is blind to another (e.g., bit-flips). This imperfection alters the causal influences involving physical qubits, but the logical-to-logical causal influence and the logical-to-ancilla remain the same.

\subsubsection{Physical-to-Ancilla Causal Influence}
For the repetition code, the physical-to-ancilla influence becomes dependent on the initial logical state:
\begin{subequations}
\begin{empheq}[box=\widefbox]{align}
    \cinf^{\mathrm{rep}}_{\phys,\anc} &= \frac{1}{3D_\anc(D_\anc^2 + 1)}\left( 1 + \frac{\braket{Z_L}^2}{2}\right)\quad &\text{(pre-measurement)}\\
    \cinf^{\mathrm{rep},\mgate}_{\phys,\anc} &= \frac{1}{6D_\anc(D_\anc^2 + 1)}\quad &\text{(post-measurement)}
\end{empheq}
\end{subequations}
The state dependence arises through the pre-measurement influence's sensitivity to $\braket{Z_L}^2$-dependent syndrome coherences. More specifically, if we compute $\tr_{\phys}(\cR'[V_j\rho V_j^\dag])$ and substitute $V_j = \cos(\alpha)\idm - i\sin(\alpha)\hat{n}\cdot\vec{\sigma}_j$, then we find state-independent diagonal components and $\braket{Z_L}^2$-dependent off-diagonal components.

\subsubsection{Physical-to-Logical Causal Influence}

Unlike weight-3 codes, the repetition code's inability to correct all errors results in a nonzero physical-to-logical influence, which is also state dependent:
\begin{equation}
    \boxed{\cinf^{\mathrm{rep}}_{\phys,L} = \frac{1}{30}\left(1 - \braket{Z_L}^2\right)\quad\text{(pre- or post-measurement)}}
\end{equation}
This influence is zero for the logical basis states $\ket{0_L}$ and $\ket{1_L}$ (where $\braket{Z_L}^2=1$), for which phase-flip perturbations are corrected. However, for superposition states, these errors corrupt the logical information, allowing information about the physical kick to leak into the logical subspace and resulting in a nonzero causal influence.

\bibliography{references}

\end{document}